\documentclass[a4paper, 11pt]{article}  
 
\usepackage{amsmath, amssymb}
\usepackage{amsthm}
\usepackage[OT1]{fontenc}

\usepackage[margin=1in]{geometry}

\usepackage{color}
\usepackage{hyperref}
\usepackage{amsmath, amsfonts, amsthm, amssymb}
\usepackage{verbatim}
\usepackage{stmaryrd}
\usepackage{fancyhdr}
\usepackage{float}
\usepackage{graphicx}
\usepackage{epstopdf}
\usepackage{eurosym}
\usepackage{enumitem}
\usepackage[nottoc,notlot,notlof]{tocbibind}
\usepackage{bbm}
\usepackage{multirow}

\newfam\msbfam

\newtheorem{theorem}{Theorem}[section]
\newtheorem{exa}{Example}[section]

\newtheorem{lemma}[theorem]{Lemma}
\newtheorem{proposition}[theorem]{Proposition}

\theoremstyle{definition}
\newtheorem{definition}[theorem]{Definition}
\newtheorem{remark}[theorem]{Remark}

\newtheorem{assumption}[theorem]{Assumption}

\numberwithin{equation}{section}

\DeclareMathOperator{\Tr}{Tr}

\begin{document} 



\setcounter{tocdepth}{4}


\title{Dynamic Portfolio Optimization with Looping Contagion Risk}

\author{
Longjie Jia\thanks{Department of Mathematics, Imperial College, London SW7 2AZ, UK (longjie.jia13@imperial.ac.uk)}
	\and 
Martijn Pistorius\thanks{Department of Mathematics, Imperial College, London SW7 2AZ, UK (m.pistorius@imperial.ac.uk)}
	\and
Harry Zheng\thanks{Corresponding author. Department of Mathematics, Imperial College, London SW7 2AZ, UK (h.zheng@imperial.ac.uk)}
}
\date{}
\maketitle

\abstract
In this paper we consider a utility maximization problem with defaultable stocks and looping contagion risk. We assume that the default intensity of one company depends on the stock prices of itself and other companies, and the default of the company induces immediate drops in the stock prices of the surviving companies. We prove that the value function is the unique viscosity solution of the HJB equation.  We also perform some numerical tests to compare and analyse the statistical distributions of the terminal wealth of log utility and power utility based on two strategies, one using the full information of intensity process and the other a proxy constant intensity process. 


 
\section{Introduction}
There has been extensive research in dynamic portfolio optimization and credit risk modelling, both in theory and applications (see Pham (2009), Brigo and Morini (2013),  and references therein).  Utility maximization with credit risk is one of the important research areas, which is to find the optimal value and optimal control in the presence of possible defaults of  underlying securities or names. The early work includes Korn and Kraft (2003) using the firm value structural approach and Hou and Jin (2002) using the reduced form intensity approach. Defaults are caused by exogenous risk factors such as correlated Brownian motions, Ornstein-Uhlenbeck or CIR intensity processes.   Bo et al. (2010) consider an infinite horizon portfolio optimization problem with a log utility and assume that both the default risk premium and the default intensity  dependant on an external factor following a diffusion process and show the pre-default value function  can be reduced to a solution of a quasilinear parabolic PDE (partial differential equation).  
Capponi and Figueroa-Lopez (2011)  assume a Markov regime switching model and derive the dynamics of the defaultable bond  and prove a verification theorem with applications to log and power utilities. Callegaro et al. (2012) consider a wealth allocation problem with several defaultable assets whose dynamics depend on a partially observed external factor process.

Contagion risk or endogenous risk has grown into a major topic of interest  as it is clear that the conventional dependence modelling of assets using  covariance matrix cannot capture the sudden market co-movements. The failure of one company will have direct impacts on the performance of other related companies. For example,
during the global financial crisis of 2007-2008, the default of Lehman Brothers led to  sharp falls in stock prices of other investment banks and stock indices such as Dow Jones US Financial Index. Since defaults are rare events, one may have to rely on the market information of other companies or indices to infer the  default probability of one specific company. For example, one can often observe in the financial market data that the stock price of one company has negative correlation with the CDS (credit default swap) spread (a proxy of default probability)  of another company. A commonly used contagion risk model  is the interacting intensity model (see  Jarrow and Yu (2001)) in which the default intensity of one name jumps whenever there are defaults of other names in a portfolio.  Contagion risk has great impact on pricing and hedging portfolio credit derivatives (see Gu et al. (2013)).

There is limited research in the literature on dynamic portfolio optimization with contagion risk.
Jiao and Pham (2011) consider a financial market with  one stock which jumps downward at the default time of a counterparty  which is not traded and not affected by the stock and, for power utility, solve  the post-default   problem by the convex duality method and show the process defined by the pre-default value function satisfies a BSDE (backward stochastic differential equations).  
Jiao and Pham (2013) discuss  multiple defaults of a portfolio  with exponential utility and prove a verification theorem for  the value function characterized by a system of BSDEs.
Bo and Capponi (2016) consider a market consisting of a risk-free bank account, a stock index, and a set of CDSs. The default of one name may trigger a jump in default intensities of other names in the portfolio, which in turn leads to jumps in the market valuation of CDSs referencing the surviving names and affects the optimal trading strategies. They solve the problem with the DPP (dynamic programming principle) and, for power utility, find the optimal trading strategy on the stock index is Merton\rq{}s strategy, and those on the CDSs can be determined by a system of recursive ODEs (ordinary differential equations). Capponi and Frei (2017) introduce an equity-credit portfolio with  a market consisting of a risk-free bank account, defaultable stocks, and CDSs  referencing these stocks. The default intensities of companies are functions of stock prices  and some  external factors, which provides a genuine looping contagion default structure. For a log utility investor, there exists an explicit optimal strategy which crucially depends on the existence of CDSs in the portfolio, see Remark \ref{CF(2007)comments} for details.

In this paper we analyse the interaction of market and credit risks and its impact on dynamic portfolio optimization. The market is assumed to have one risk-free savings account, and multiple defaultable stocks in which the underlying companies may default and the value of defaulted stock price becomes zero.   The default time of any stock is the first jump time of a pure jump process driven by an intensity process that depends on all the surviving stock prices, and the surviving stock prices jump at time of default. This setup characterizes an investment with multiple stocks that are closely dependent on each other, both endogenously and exogenously. Compared with exogenous factor models in the literature, which strongly depend on the historical calibration of factor parameters, the looping contagion model has the ability to adjust trading strategies automatically based on current stock prices in the portfolio. We  study  a  terminal wealth utility maximization  problem with general utility functions  under this  looping contagion framework.

The aforementioned papers by Bo and Capponi (2016)  and Capponi and Frei (2017)  characterize the value function as a solution of the HJB (Hamilton-Jacobi-Bellman) equation and,  for power and log utility respectively, find the optimal trading strategies with some implicit unknown functions.  For general utilities, it is essentially impossible one may  guess a solution form of the HJB equation nor can one apply the verification theorem. In that case, a standard approach to studying the value function is the viscosity solution method. We prove, in addition to the verification theorem, that
 the value function is the unique viscosity solution of the HJB equation.  The result is important  as it lays a solid theoretical foundation for numerical schemes to find the value function,  in contrast to the verification theorem that requires priori the existence of a classical solution to the HJB equation, which is  in general difficult to prove. To the best of our knowledge, this is the first time the viscosity solution properties of the value function are studied and established in the literature of utility maximization with looping contagion risk. This is one of the main contributions of the paper.

We perform some numerical and robust tests to compare the statistical distributions of terminal wealth of log utility and power utility based on two trading strategies,  one uses the full information of intensity process, the other a proxy constant intensity process. These two strategies may be considered respectively the active and passive  portfolio investment. The numerical examples show that, statistically, they have  similar terminal wealth distributions, but active portfolio investment is more volatile in general.  Furthermore,  we illustrate the financial insight of the looping contagion model via a similar numerical test, but with different initial stock prices. The numerical test assumes that the constant intensity is estimated from historical calibration window, but there are big falls of stock prices at the start of the investment. The numerical example shows that the terminal wealth based on strategies using stock dependent intensity would have  much higher expected return and standard deviation than the one using a constant intensity. Therefore,  one may greatly improve the performance of investment  if one uses the information of stock dependant default intensity in a financial crisis period.

The rest of the paper is organized as follows. In Section 2 we introduce  the market model and state the main results, including the continuity of the value function for one-sided contagion case (Theorem \ref{continuity property of value function}), the verification theorem (Theorem \ref{pre-default verification theorem}), and
the unique viscosity solution property of the value function (Theorems \ref{vis} and  \ref{vis_unique}).
In Section 3 we perform numerical and robust tests with statistical distribution analysis for log and  power utility.  In Section 4 we prove Theorems 
\ref{continuity property of value function}, \ref{pre-default verification theorem}, \ref{vis} and \ref{vis_unique}.
  Section 5   concludes the paper.

\section{Model Setting and Main Results}
Let $(\Omega,\mathcal{G},(\mathcal{G}_t)_{t\geq 0},\mathbb{P})$ be a complete probability space satisfying the usual conditions and $(\mathcal{G}_t)_{t\geq 0}$ a filtration to be specified below. 
Let the market consist of one risk-free bank account with value process $(B_t)_{t \geq 0}$ and interest rate $r$ and $N$ defaultable stocks with price process $(S_t)_{t \geq 0}:=(S_t^1,...,S_t^N)^T_{t\geq 0}$, where $a^T$ is the transpose of a vector $a$. Let $(\mathcal{F}_t)_{t\geq 0}$ be the filtration generated by $N$ correlated Brownian motions $(W_t)_{t\geq 0}:=(W^1_t,...,W^N_t)^T_{t\geq 0}$, which represents the market information. 
Let $\tau:=(\tau_1,...,\tau_N)$ be a vector of nonnegative random variables representing the default time of each defaultable stock, defined by
\[
\tau_i:=\inf \left\{s\geq t: \int_t^s h^i_u du \geq \mathcal{X}_i \right\},
\]
where $(h^i_t)_{t\geq 0}$ is an intensity rate process and $\mathcal{X}_i$ is a standard exponential variable on the probability space
$(\Omega,\mathcal{G},\mathbb{P})$ and is independent of the filtration $(\mathcal{F}_t)_{t\geq 0}$, which means that $\tau_i$ is a totally inaccessible stopping time. We make the further assumption that $\mathcal{X}_i$ is independent of $\mathcal{X}_j$ for $i\neq j$. Under this assumption, the default of each stock is independent. 

Let $(\mathcal{H}_t)_{t\geq 0}$ be the filtration generated by the default indicator process $(\mathbb{H}_t)_{t\geq 0}:=(H^1_t,...,H^N_t)^T_{t\geq 0}$ where each of the default process $H_t^i$  is associated with the intensity process $(h^i_t)_{t \geq 0}$ and defined by
$H_t^i := \mathbb{I}_{\{ \tau_i\leq t\} }$, the indicator function that equals 0 if $\tau_i>t$ and 1 otherwise. Denote the value of indicator process $\mathbb{H}_t$ by $z$, thus $z\in I:=\{0,1\}^N$. The indicator process $\mathbb{H}_t$ can only jump from $z:=(z_1,...z_N)^T$ to its neighbor state $z^i:=(z_1,...,1-z_i,...,z_N)^T$ with rate $(1-z_i)h^i_t$ for $i\in \{1,...,N\}$. We denote $N_z$ the number of surviving stocks when $\mathbb{H}_t=z$ and $I_z$ the set of surviving stock numbers.

Finally, let 
$(\mathcal{G}_t)_{t\geq 0}$ be an enlarged filtration, defined by $\mathcal{G}_t = \mathcal{F}_t \lor \mathcal{H}_t$, which contains both the market information and the default information. 
The stopping time $\tau_i$ defined in above way satisfies the so-called $H$-hypothesis, which means any $\mathcal{F}$-square integrable martingale is also a $\mathcal{G}$-square integrable martingale (see Bielecki and Rutkowski (2003)), a property we will use later in the proofs. The market model is driven by the following stochastic differential equations (SDEs):
$$
\frac{dS^i_t}{S^i_{t-}} = \mu_i dt+\sigma_i dW^i_t-L_i^T d\mathbb{H}_t,
\quad
\frac{dB_t}{B_{t}} = r dt,
$$
for integer $i\in \{1,...,N\}$ where $\mu_i$ is the growth rates of  $S^i$, respectively, $\sigma_i$ is the volatility rate. The vector $L_i:=(L_{i1},...,L_{iN})^T$ represents the default impact of each stock to the $i$th stock, thus $L_{ii}=1$. 

All coefficients are positive constants to simplify discussions. We assume that the defaults of stocks do not occur at the same time. At default time $\tau_i$ the defaultable stock price $S^i$ falls to zero and the other stock price $S^j$ is reduced by a percentage of $L_{ji}$ for $i\neq j$.  We require that $L_{ii}=1$ and $L_{ij}<1$ for $i\neq j$. $L_{ji}<1$ ensures the other stock  price $S^j$ does not fall to zero at default time of $\tau_i$.  We denote  by $K$ a generic constant which may have different values at different places. 

\begin{assumption} \label{h assumption}
The intensity process $(h_t^i)_{t\geq 0}$ of the default indicator process $(H_t^i)_{t\geq 0}$ can be represented by $h_t^i=h(S_{t-}^z,z)$, a function of surviving stock prices $S_{t-}^z:=(S_{t-}^i)_{i\in I_z}$ and the state of default indicator process $\mathbb{H}_{t-}=z$. For simplicity, we denote $h(S_{t-}^z,z)$ by $h_z^i(S_{t-})$. We further assume that $h_z^i$ is bounded and continuous in $S_{t-}^z$ for $\forall z\in I$ and $i\in \{1,\ldots,N\}$.
\end{assumption}

To classify the looping contagion model setting, we give two examples which contain only two stocks in the market, denoted by $(S_t)_{t\geq 0}$ and $(P_t)_{t\geq 0}$.

\begin{exa} \label{example1}
(\textbf{One-sided contagion}) In this case, $(S_t)_{t\geq 0}$ denotes the price of ETF (exchange-traded-fund) on DJ US Financial Index and $(P_t)_{t\geq 0}$ denotes the price of a US investment bank. We may treat the ETF as default-free and its stock price reflects the whole US banking industry and thus has impact on the performance of the individual bank. Then the model is given by
$$
\frac{dS_t}{S_{t-}} = \mu^Sdt+\sigma^SdW^S_t-L^SdH_t,
\quad
\frac{dP_t}{P_{t-}} = \mu^Pdt+\sigma^PdW^P_t-dH_t,
$$
where $\mu^S$ and $\mu^P$ are  the growth rates of  $S$ and $P$, respectively, $\sigma^S$ and $\sigma^P$  are the volatility rates, and $L^S<1$ is the percentage loss of the stock  $S$ upon the default of stock $P$. At default time $\tau$ the defaultable stock price $P$ falls to zero and the stock  price $S$ is reduced by a percentage of $L^S$.  The intensity process $(h_t)_{t\geq 0}$ of the default indicator process $(H_t)_{t\geq 0}$ can be represented by $h_t=h(S_{t-},P_{t-})$. 
\end{exa}

\begin{exa} \label{example2}
(\textbf{Looping contagion}) In this case, both $(S_t)_{t\geq 0}$ and $(P_t)_{t\geq 0}$ denote the prices of single stocks. Then the model is given by
$$
\frac{dS_t}{S_{t-}} = \mu^Sdt+\sigma^SdW^S_t-dH^S_t-L^SdH^P_t,
\quad
\frac{dP_t}{P_{t-}} = \mu^Pdt+\sigma^PdW^P_t-L^PdH^S_t-dH^P_t.
$$
At default time of $S$ (resp. $P$), the stock price $S$ (resp. $P$) falls to zero and the stock  price $P$ (resp. $S$) is reduced by a percentage of $L^P$ (resp. $L^S$).  The intensity process $h_{(0,0)}^S(t)$ (resp. $h_{(0,0)}^P(t)$) of the default indicator process $(H^S_t)_{t\geq 0}$ (resp. $(H^P_t)_{t\geq 0}$) can be represented by $h_{(0,0)}^S(t)=h_{(0,0)}^S(S_{t-},P_{t-})$ (resp. $h_{(0,0)}^P(t)=h_{(0,0)}^P(S_{t-},P_{t-})$). After the default of $S$ (resp. $P$), the intensity process $h_{(1,0)}^P(t)$ (resp. $h_{(0,1)}^S(t)$) of the default indicator process $(H^P_t)_{t\geq 0}$ (resp. $(H^S_t)_{t\geq 0}$) can be represented by $h_{(1,0)}^P(t)=h_{(1,0)}^P(P_{t-})$ (resp. $h_{(0,1)}^S(t)=h_{(0,1)}^S(S_{t-})$).
\end{exa}

An investor dynamically allocates proportions $(\pi^1, \ldots, \pi^N, 1-\sum_{i=1}^N \pi^i)$ of the total wealth into the stocks and the bank account. 
The admissible control set $\mathcal{A}$ is the set of control processes $\pi$ that are progressively measurable with respect to the filtration $(\mathcal{G}_t)$  and $\pi_t\in A$ for all $t\in[0,T]$. The set $A$ is defined by
\[
A:= \left\{\pi\in O \mbox{ and }1-\sum_{i=1}^N L_{ij}\pi^i\geq \epsilon_A  \mbox{ for } \forall j\in \{1,...,N\}
\right \},
\]
where $O$ is a bounded set in $\mathbb{R}^N$ and  $\epsilon_A$ is a positive constant.
The dynamics of the wealth process $(X_t)_{t \geq 0}$ is given by
\begin{equation} \label{wealth process}
\begin{split}
\frac{dX_t}{X_{t-}} = \left(r+\pi_{t}^TD_t\theta\right)dt + \pi_{t}^TD_t\sigma dW_t - \pi_{t-}^TD_tLd\mathbb{H}_t,
\end{split}
\end{equation} 
where 
$$ D_t:=
 \begin{pmatrix}
1-H_t^1 & \ldots & 0 \\
\vdots  & \vdots  & \vdots \\
0       & \ldots & 1-H_t^N
   \end{pmatrix},\;
\theta: =  \begin{pmatrix}  \mu_1-r\\  \vdots \\ \mu_N-r \end{pmatrix}, \;
\sigma:=
\begin{pmatrix}
\sigma_1 & \ldots & 0 \\
\vdots  & \vdots  & \vdots \\
0       & \ldots & \sigma_N
   \end{pmatrix},\;
L:=
 \begin{pmatrix}
L_{11} & \ldots & L_{1N} \\
\vdots  & \vdots  & \vdots \\
L_{N1}  & \ldots & L_{NN}
   \end{pmatrix}.   
$$
The matrix-valued process $(D_t)_{t\geq 0}$ is adapted to the filtration $(\mathcal{H}_t)_{t\geq0}$ and plays the role of removing the defaulted stocks. 
Even though the admissible control set is still $A$ after default time $\tau_i$,  $\pi^i_t=0$ and is not  a  variable but a constant.   The requirement  $1-\sum_{i=1}^N L_{ij}\pi^i\geq \epsilon_A$ for $\forall j\in \{1,...,N\}$ ensures that when $j$th stock defaults, the maximum percentage loss of the wealth does not exceed $1- \epsilon_A$, in other words, if $x$ is the pre-default wealth, then the post-default wealth is at least $ \epsilon_A x$. 

\begin{remark} \label{Proposition of wealth process}
For a given control process $\pi\in \mathcal{A}$, equation (\ref{wealth process}) admits a unique strong solution that satisfies
\begin{equation} \label{strong solution property}
\sup_{t\in[0,T]}\mathbb{E}\left[X_t^\alpha\right] \leq K x^\alpha
\end{equation}
for any $\alpha>0$. This can be easily verified 
as  $X_t^\alpha = x^\alpha N_t M_t$, where
\begin{eqnarray*}
N_t &:= &\exp \left(  \alpha\int_0^t\left(r+\pi_u^TD_u\theta \right) du + {1\over 2}(\alpha^2-\alpha)\int_0^t \pi_u^TD_u\Sigma D_u\pi_u du +\alpha \sum_{j=1}^N\int_0^t \ln\left(1-\sum_{i=1}^N L_{ij}\pi_{u-}^i\right) dH^j_u \right),\\
M_t&:=&\exp\left(\int_0^t \alpha\pi_u^TD_u\sigma dW_u - \frac{1}{2}\alpha^2\int_0^t \pi_u^TD_u\Sigma D_u\pi_u du \right),\\
\Sigma&:=&
 \begin{pmatrix}
(\sigma_1)^2 & \rho_{12}  \sigma_1\sigma_2 & \ldots & \rho_{1N}\sigma_1\sigma_N \\
\vdots & \vdots & \vdots & \vdots \\
  \rho_{N1} \sigma_1\sigma_N & \rho_{N2}\sigma_2\sigma_N & \ldots & (\sigma_N)^2
 \end{pmatrix}, \\
 \pi_u &:= & (\pi^1_u, \ldots, \pi^N_u)^T.
\end{eqnarray*}
Note that $\rho_{ij}$ is the correlation between Brownian motion $W^i$ and $W^j$. Since $A$ is a bounded set and $1-\sum_{i=1}^N L_{ij}\pi^i\geq\epsilon_A$ for $\forall j\in \{1,...,N\}$, we have $|N_t|<K$,  independent of $t$, and $M_t$ is an exponential martingale, thus $\mathbb{E}\left[M_t\right]=1$, which gives (\ref{strong solution property}). 
\end{remark}

Our objective is to maximize the expected utility of the terminal wealth, that is, 
$$ \sup_{\pi\in \mathcal{A}} \mathbb{E}[U(X_T^\pi)],$$
where $U$ is a utility function defined on $[0,\infty)$ and satisfies the following assumption.
\begin{assumption} \label{utility assumption}
The utility function $U$ is continuous, non-decreasing, concave, and satisfies $U(0)>-\infty$ and 
$\left|U(x)\right| \leq K\left(1+x^\gamma\right)$ for all $x\in [0,\infty)$, where $K>0$ and $0<\gamma<1$ are constants. 
\end{assumption}

Depending on the default scenarios, the value function is defined by
\begin{equation*} \label{value function}
v_z(t,x,s) = \sup_{\pi\in \mathcal{A}} \mathbb{E}\left[U(X_T^\pi)|X_t=x, S_t=s,\mathbb{H}_t=z\right]
\end{equation*} 
for $(t,x,s)\in [0,T]\times (0,\infty)^{N_z+1}$ and $z\in I$. Note that if  $h$ is independent of $s$, then the value function $v_z$ is a function of $t, x$ only.

\begin{remark}
Combining Assumption \ref{utility assumption} and Remark \ref{Proposition of wealth process}, we have $|v_z(t,x,s)|\leq K(1+x^\gamma)$.
\end{remark}

For the one-sided contagion model defined in Example \ref{example1}, the problem can be naturally split into pre-default case and post-default case. The latter is a standard utility maximization problem as stock $P$ disappears and the post-default value function $v_1$ is a function of time $t$ and wealth $x$ only,  see Pham (2009). We have the following continuity result for the pre-default value function $v_0$.

\begin{theorem}  \label{continuity property of value function}
For the one-sided contagion model (Example \ref{example1}), assume further that $h$ is non-increasing in $p$, monotone in $s$ and Lipschitz continuous in $s, p$, and $U$ satisfies $\left|U(x_1)-U(x_2)\right| \leq K\left|x_1-x_2\right|^\gamma$ for all $x_1,x_2\in [0,\infty)$. Then the pre-default value function $v_0$ is continuous in $(t,x,s,p)\in [0,T]\times [0,\infty)\times (0,\infty)^2$.
\end{theorem}

\begin{remark}
We assume $h$ is non-increasing in $p$ as intuitively the default probability of one company is non-increasing with its own stock price. We also assume that $h$ is monotone in $s$ as we consider $S$ and $P$ are strongly correlated in the sense that the default probability of stock $P$ is either positively or negatively affected by the  stock  $S$.
The continuity of pre-default value function for the one-sided contagion model relies on the special structure that there is only one default process in the place. For general looping contagion models,  the continuity of the value function is difficult to obtain as the order of multiple jumps is random. 
\end{remark}

Applying the DPP, one can show that the value function satisfies the following HJB equation:
\begin{equation} \label{General pre-default HJB equation}
-\sup_{\pi\in A}\mathcal{L}^\pi w_z(t,x,s)=0
\end{equation}
for $(t,x,s)\in [0,T)\times (0,\infty)^{N_z+1}$ and $z\in I$ with terminal condition $w_z(T,x,s)=U(x)$, where 
$\mathcal{L}^\pi$ is the infinitesimal generator of processes $S$, $\mathbb{H}$ and $X$  with control $\pi$,   given by
\begin{eqnarray}
\mathcal{L}^\pi w_z(t,x,s) &= &\frac{\partial w_z}{\partial t} + (r+\theta^T\pi)x \frac{\partial w_z}{\partial x}  + \sum_{i\in I_z}\mu_is_i \frac{\partial w_z}{\partial s_i} + \frac{1}{2}\pi^T\Sigma\pi x^2 \frac{\partial^2 w_z}{\partial x^2} + \frac{1}{2}\sum_{i\in I_z}\sigma_i^2s_i^2 \frac{\partial^2 w_z}{\partial s_i^2}   \nonumber \\
&& {} + \sum_{i,j\in I_z,i<j} \rho_{ij}\sigma_i\sigma_js_is_j \frac{\partial^2 w_z}{\partial s_i\partial s_j} + \sum_{i\in I_z}\rho_i^T\sigma\pi\sigma_i xs_i \frac{\partial^2 w_z}{\partial x\partial s_i} \nonumber \\
&& {} + \sum_{i\in I_z} h^i_z(s) \left(w_{z^i}\left(t,x\left(1-\sum_{j=1}^N L_{ji}\pi^j\right),s^i\right)-  w_z\right), \label{e2.4}
\end{eqnarray}
where $s^i:=\left(s_1(1-L_{1i}),\ldots,s_j(1-L_{ji}),\ldots,s_N(1-L_{Ni})\right)^T$ for $j\in I_{z^i}$ and $\rho_i:=\left(\rho_{i1},\ldots,\rho_{ij},\ldots,\rho_{iN}\right)^T$ for $j\in I_{z}$. Note that the dimension of $s^i$ is $N_{z^i}$ which is equal to $N_z-1$ as we have removed the $i$th defaulted stock.

We next give a verification theorem for the value function.
\begin{theorem} \label{pre-default verification theorem}
Assume  that the function tuple $w:=(w_z)_{z\in I}$ where $w_z\in C\big([0,T]\times (0,\infty)^{N_z+1}\big)  \cap C^{1,2,\ldots,2}\big([0,T)\times (0,\infty)^{N_z+1}\big)$ for any $z\in I$ solves (\ref{General pre-default HJB equation}) with the terminal condition $w_z(T,x,s)=U(x)$,
that $w_z$  satisfies a growth condition
$ |w_z(t,x,s)| \leq K\left(1+x^\gamma\right)$
for $0<\gamma<1$, that the maximum of the Hamiltonian in (\ref{General pre-default HJB equation}) is achieved at 
 $\widehat{\pi} (t, x, s, z)$ in $A$,  and that SDE (\ref{wealth process}) admits a unique strong solution  $X_t^{\widehat{\pi}}$ with control $\widehat{\pi}$.
Then $w_z$ coincides with the value function  $v_z$ and $\widehat{\pi}$ is the optimal control process.
\end{theorem}

\begin{remark}
For log utility $U(x)=\ln x$, the assumption $U(0)>-\infty$ is not satisfied. However, one may postulate that  the value function has a form 
$w_z(t,x,s)=\ln x + f_z(t,s)$, where $f$ is a solution of a linear PDE, see (\ref{Log utility HJB equation}). If we assume $f_z\in C\left([0,T]\times (0,\infty)^{N_z+1}\right)\cap C^{1,2,...,2}\left([0,T)\times (0,\infty)^{N_z+1}\right)$ and is bounded, then one can show that
$w_z$ is indeed the value function with 
the same proof as that of Theorem~\ref{pre-default verification theorem}   except one change: instead of using $ |w_z(t,x,s)| \leq K\left(1+x^\gamma\right)$, which does not hold for log utility, one uses $ |w_z(t,x,s)| \leq K\left(1+|\ln x|\right)$. Since 
$$ \ln X_u =\ln x +  \int_t^u\Big(r+\pi_{\bar u}^T D_{\bar u} \theta-\frac{1}{2}\pi_{\bar u}^TD_{\bar u}\Sigma D_{\bar u} \pi_{\bar u} \Big) d{\bar u} + \int_t^u \pi_{\bar u}^TD_{\bar u}\sigma dW_{\bar u}  + \sum_{j=1}^N\int_t^u \ln\left(1-\sum_{i=1}^N L_{ij}\pi_{\bar{u}-}^i\right) dH^j_{\bar{u}}
$$
for $u\in [t,T]$, we have  $\mathbb{E}\left[\left|\ln {X_u}\right|^2\right]
\leq K(1+(\ln x)^2)$, which provides the required uniform integrability property  in the proof. 
\end{remark}

The verification theorem assumes  the existence of a classical solution of the HJB equation (\ref{General pre-default HJB equation}), which may not be true for the value function $v_z$. Next we show that the value functions $\{v_z\}_{z\in I}$ is the unique viscosity solution to the PDE system characterized by (\ref{General pre-default HJB equation}) based on the following definition. 

To facilitate discussions of viscosity solution, we define $F$ function by
\[
F_z\left(t,x,s,w,\nabla_{(t,x,s)}w_z,\nabla^2_{(x,s)}w_z\right) = - \sup_{\pi\in A} \mathcal{L}^{\pi}w_z(t,x,s),
\]
where $\nabla_{(t,x,s)}w_z\in \mathbb{R}^{N_z+2}$ is the gradient vector of $w_z$ with respect to $(t,x,s)$, and $\nabla^2_{(x,s)}w_z\in \mathbb{R}^{(N_z+1)\times (N_z+1)}$ is the Hessian matrix of $w_z$ with respect to $(x,s)$. $w_z$ and its derivatives are evaluated at $(t,x,s)$. The HJB equation (\ref{General pre-default HJB equation}) is the same as
\[
F_z\left(t,x,s,v,\nabla_{(t,x,s)}v_z,\nabla^2_{(x,s)}v_z\right) = 0
\]
for $\forall z\in I$.

\begin{definition} \label{viscosity}
(i) $w:=(w_z)_{z\in I}$ is a  viscosity subsolution  of the PDE system (\ref{General pre-default HJB equation}) on $[0,T)\times (0,\infty)^{N+1}$ if
\[
F_{\bar{z}}\left(\bar{t},\bar{x},\bar{s},\varphi,\nabla_{(t,x,s)}\varphi_{\bar{z}},\nabla^2_{(x,s)}\varphi_{\bar{z}}\right ) \leq  0
\]
for all $\bar z\in I$, $(\bar{t},\bar{x},\bar{s})\in [0,T)\times (0,\infty)^{N_{\bar z}+1}$ and testing functions $\varphi:=(\varphi_z)_{z\in I}\in C^{1,2,...,2}\left([0,T)\times (0,\infty)^{N_z+1}\right)$ such that $(w_{\bar{z}})^{*}(\bar{t},\bar{x},\bar{s})=\varphi_{\bar{z}}(\bar{t},\bar{x},\bar{s})$ and $(w_{z})^{*}\leq \varphi_{z}$ for $\forall z\in I$ on $[0,T)\times (0,\infty)^{N_z+1}$, where
$(w_z)^{*}$ is the upper-semicontinuous envelope  of $w_z$, defined by $(w_z)^*(\bar t,\bar x,\bar s)=\limsup_{(t,x,s)\rightarrow (\bar t,\bar x,\bar s)}w_z(t,x,s)$.

(ii) $w:=(w_z)_{z\in\{0,1\}^N}$ is a viscosity supersolution of the PDE system (\ref{General pre-default HJB equation}) on $[0,T)\times (0,\infty)^{N+1}$ if
\[
F_{\bar{z}}\left(\bar{t},\bar{x},\bar{s},\varphi,\nabla_{(t,x,s)}\varphi_{\bar{z}},\nabla^2_{(x,s)}\varphi_{\bar{z}}\right ) \geq 0
\]
for all $\bar z\in I$, $(\bar{t},\bar{x},\bar{s})\in [0,T)\times (0,\infty)^{N_{\bar z}+1}$ and testing functions $\varphi:=(\varphi_z)_{z\in I}\in C^{1,2,...,2}\left([0,T)\times (0,\infty)^{N_z+1}\right)$ such that $(w_{\bar{z}})_{*}(\bar{t},\bar{x},\bar{s})=\varphi_{\bar{z}}(\bar{t},\bar{x},\bar{s})$ and $(w_{z})_{*}\geq \varphi_{z}$ for $\forall z\in I$ on $[0,T)\times (0,\infty)^{N_z+1}$, 
, where
$(w_z)_{*}$ is the lower-semicontinuous envelope  of $w_z$, defined by $(w_z)_*(\bar t,\bar x,\bar s)=\liminf_{(t,x,s)\rightarrow (\bar t,\bar x,\bar s)}w_z(t,x,s)$.

(iii) We say that $w$ is a viscosity solution of the PDE system (\ref{General pre-default HJB equation}) on $[0,T)\times (0,\infty)^{N+1}$ if it is both a viscosity subsolution and supersolution of (\ref{General pre-default HJB equation}).
\end{definition}

Based on the above definition, we have the following viscosity solution property for the value function.
\begin{theorem}\label{vis}
The value function $v=(v_z)_{z\in I}$ is 
a viscosity solution of the PDE system (\ref{General pre-default HJB equation}) on $[0,T)\times (0,\infty)^{N+1}$, satisfying the growth condition $|v_z(t,x,s)|\leq K(1+x^\gamma)$ for some constant $0<\gamma<1$. 
\end{theorem}

To prove the uniqueness of the viscosity solution, we need to introduce a structure condition on the model.
\begin{assumption} \label{structure}
The following inequality holds:
 $$J_z(\pi)\leq \frac{K}{\epsilon}\left(|x_1-x_2|^2 + \sum_{i\in I_{z}} |s_{1i}-s_{2i}|^2 \right), \
 \forall \pi\in A, z\in I, $$  
 where
\begin{eqnarray}
J_z(\pi) &:= & \frac{1}{2}\pi^T\Sigma\pi \left(x_1^2 Q_{1,1} - x_2^2 Q_{1,1}^{\prime}\right) + \frac{1}{2}\sum_{i\in I_{z}} \sigma_i^2\left(s_{1i}^2 Q_{k_i,k_i} - s_{2i}^2 Q_{k_i,k_i}^{\prime}\right) \nonumber \\
&& {} + \sum_{i,j\in I_{z},i<j} \rho_{ij}\sigma_i\sigma_j\left(s_{1i}s_{1j}Q_{k_i,k_j} - s_{2i}s_{2j} Q_{k_i,k_j}^{\prime}\right)  + \sum_{i\in I_{z}}\rho_i^T\sigma\pi\sigma_i \left(x_1s_{1i} Q_{1,k_i} - x_2s_{2i} Q_{1,k_i}^{\prime}\right) \nonumber 
\end{eqnarray} 
and matrices $Q$ and $Q^{\prime}$ satisfy
\begin{equation*}
 \begin{pmatrix}
  Q & 0 \\
  0 & -Q^{\prime}
 \end{pmatrix}
\leq \frac{3}{\epsilon}
 \begin{pmatrix}
  I_{N_{z}+1} & -I_{N_{z}+1} \\
  -I_{N_{z}+1} & I_{N_{z}+1}
 \end{pmatrix}.
\end{equation*} 
\end{assumption}

\begin{remark}
The dimension of matrices $Q$ and $Q^{\prime}$ in Theorem \ref{vis_unique}  is $N_z+1$. We use $k_i$ to represent the right index of matrices which corresponds to $s_i$ where $i\in I_z$. The introduction of $k_i$ is to resolve the gap between $s$ index and matrix index. We use a simple example to illustrate the definition of $k_i$. For example, $I_z:=\{3,4,6\}$. In this case, there are three surviving stocks in the market, namely $s_3,s_4,s_6$. The dimension of matrices $Q$ and $Q^{\prime}$ is 4 (including 3 surviving stocks and the wealth process $x$). Then $k_3=2,k_4=3,k_6=4$.
\end{remark}

\begin{remark}
For the simplest case where there are only two defaultable stocks in the market, e.g. Example \ref{example1} and Example \ref{example2}, Assumption \ref{structure} holds for $\forall \rho\in(-1,1)$, as $J_z(\pi)$ can be written as
\[
\begin{split}
J_z(\pi) &= \frac{1}{2} \xi^T \begin{pmatrix}
  Q & 0 \\
  0 & -Q^{\prime}
\end{pmatrix}
\xi
 + \frac{1}{2}(1-\rho^2)(\sigma^P)^2 \big(\pi^Px_1, 0, p_1,  \pi^Px_2, 0, p_2 \big) 
\begin{pmatrix}
  Q & 0 \\
  0 & -Q^{\prime}
\end{pmatrix}
 \big(\pi^Px_1, 0, p_1,  \pi^Px_2, 0, p_2 \big)^T, \\
\end{split}
\]
where $\xi=\big(m^T \pi x_1, \sigma^Ss_1, \rho \sigma^Pp_1,   m^T \pi x_2,  \sigma^Ss_2, \rho\sigma^Pp_2 \big)^T$, $m=(\sigma^S, \rho  \sigma^P)^T$, $n=(\rho\sigma^S, \sigma^P)^T$ and $\pi:=(\pi^S,\pi^P)^T$. Using the matrix inequality and  simple algebraic calculation, one can show that
$$
J_z(\pi) \leq \frac{3}{2\epsilon}\big( (m^T\pi)^2 |x_1-x_2|^2 + (\sigma^S)^2 |s_1-s_2|^2 + \rho^2(\sigma^P)^2 |p_1-p_2|^2 \big) + \frac{3}{2\epsilon}(1-\rho^2)(\sigma^P)^2\left((\pi^P)^2 |x_1-x_2|^2 + |p_1-p_2|^2\right).
$$
By the boundedness of control set $A$, Assumption \ref{structure}  holds for all $\rho\in (-1,1)$. 
\end{remark}

The next result states the uniqueness of the viscosity solution. 

\begin{theorem}\label{vis_unique} Let Assumption \ref{structure}  hold. Assume 
the value function $v=(v_z)_{z\in I_z}$, satisfies the terminal condition $v_z(T-,x,s)=U(x)$ and the boundary conditions $(v_z)^*(t,x,s)= (v_z)_*(t,x,s)$ for $(x,s)$ on the boundary of $[0,\infty)^{N_z+1}$. Then 
$v$ is the unique viscosity solution of the PDE system (\ref{General pre-default HJB equation}) on $[0,T)\times (0,\infty)^{N+1}$.
\end{theorem}

\begin{remark}
The condition  $(v_z)^*(t,x,s)=(v_z)_*(t,x,s)$ for $(x,s)$ on the boundary  is equivalent to the existence of the limit of the value function $v_z$ at boundary points. This condition is needed as the domain of  $(x,s)$ variables is $(0,\infty)^{N_z+1}$, not $(-\infty,\infty)^{N_z+1}$, in which case one may impose some polynomial growth conditions on~$v_z$, see Pham (2009), Remark 4.4.8, for further discussions on this point.
\end{remark}

\section{ Numerical Tests} 
In this section, we perform some statistical and robust  tests for  log and power utilities.  We assume that there are two defaultable stocks and one risk-free bank account in the market (Example \ref{example2}).

\subsection{Optimal  strategies for log utility} \label{section_log}
For  $U(x) = \ln x$, the post-default case $z=(1,1)$ is investing into the risk-free bank account, thus $\pi^S=\pi^P=0$ and $v_{(1,1)}(t,x)=\ln x + r(T-t)$. We conjecture that the pre-default value function $v_{(0,0)}(t,x,s,p)$ takes the form
\begin{equation} \label{value function of log utility}
v_{(0,0)}(t,x,s,p) = \ln x + f_{(0,0)}(t,s,p),
\end{equation}
and the value function $v_{(1,0)}(t,x,p), v_{(0,1)}(t,x,s)$ respectively take the forms
\begin{equation} \label{value function of log utility for other z}
v_{(1,0)}(t,x,p) = \ln x + f_{(1,0)}(t,p), \quad v_{(0,1)}(t,x,s) = \ln x + f_{(0,1)}(t,s).
\end{equation}
Substituting (\ref{value function of log utility}) and (\ref{value function of log utility for other z}) into (\ref{General pre-default HJB equation}), we get a linear PDE for $f_{(0,0)}$ depending on the value of $f_{(1,0)}$ and $f_{(0,1)}$:
\begin{equation} \label{Log utility HJB equation}
\begin{split}
& \frac{\partial f_{(0,0)}}{\partial t} + b^T(s,p) \mathcal{D}f_{(0,0)} + \frac{1}{2}\Tr\left(\sigma\sigma^T(s,p)\mathcal{D}^2f_{(0,0)}\right) - \left(h_{(0,0)}^S(s,p) + h_{(0,0)}^P(s,p)\right)f_{(0,0)}(t,s,p) \\
&  + h_{(0,0)}^S(s,p)f_{(1,0)}(t,p(1-L^P)) + h_{(0,0)}^P(s,p)f_{(0,1)}(t,s(1-L^S)) + r + \sup_{\pi\in A}G_{(0,0)}(s,p,\pi) = 0
\end{split}
\end{equation}
with the terminal condition $f_{(0,0)}(T,s,p) = 0$, where $G_{(0,0)}$ is defined by
\begin{equation*} \label{G function of log utility}
G_{(0,0)}(s,p,\pi): = -\frac{1}{2}\pi^T\Sigma\pi + \theta^T\pi + h_{(0,0)}^S(s,p)\ln(1-\pi^S-L^P\pi^P) + h_{(0,0)}^P(s,p)\ln(1-L^S\pi^S-\pi^P),
\end{equation*}
and the other notations are given by
$$
b(s,p): =  \begin{pmatrix} \mu^Ss\\   \mu^Pp \end{pmatrix}, \;
\mathcal{D}f_{(0,0)}: =  \begin{pmatrix}  \frac{\partial f_{(0,0)}}{\partial s}\\   \frac{\partial f_{(0,0)}}{\partial p} \end{pmatrix}, \;
\sigma(s,p):=
 \begin{pmatrix}
  \sigma^Ss & 0 \\
 \rho\sigma^Pp& \sqrt{1-\rho^2}\sigma^Pp
   \end{pmatrix},\;
\mathcal{D}^2f_{(0,0)}:=
 \begin{pmatrix}
  \frac{\partial^2 f_{(0,0)}}{\partial s^2} & \frac{\partial^2 f_{(0,0)}}{\partial s\partial p} \\
  \frac{\partial^2 f_{(0,0)}}{\partial s\partial p} & \frac{\partial^2 f_{(0,0)}}{\partial p^2}
 \end{pmatrix}.  
$$

By the same argument, we get a linear PDE for $f_{(1,0)}$:
\begin{equation*} \label{Log utility HJB equation}
\frac{\partial f_{(1,0)}}{\partial t} + \mu^Pp \frac{\partial f_{(1,0)}}{\partial p} + \frac{1}{2}(\sigma^P)^2p^2 \frac{\partial^2 f_{(1,0)}}{\partial p^2} - h_{(1,0)}^P(p)f_{(1,0)}(t,p) + r + h_{(1,0)}^P(p)r(T-t) + \sup_{\pi\in A}G_{(1,0)}(p,\pi) = 0
\end{equation*}
with the terminal condition $f_{(1,0)}(T,p) = 0$, where $G_{(1,0)}$ is defined by
\begin{equation*} \label{G function of log utility}
G_{(1,0)}(p,\pi): = -\frac{1}{2}(\sigma^P)^2(\pi^P)^2 + (\mu^P-r)\pi^P  + h_{(1,0)}^P(p)\ln (1-\pi^P).
\end{equation*}
The PDE associated with $f_{(0,1)}$ can be obtained similarly.

Assume the control constraint set $A$ is given by
\[
A:= \big\{\pi \mid a^S \leq \pi^S \leq b^S \mbox{ and } a^P \leq \pi^P \leq b^P
\big \},
\]
where $a^S, b^S, a^P, b^P\in \mathbb{R}$ are chosen such that $1-L^T\pi\geq \epsilon_A$ for $\forall \pi\in A$. We need to solve a constrained optimization problem:
$$ \max_{\pi\in A} G_{(0,0)}(s,p,\pi).
$$
Since $A$ is compact and $G_{(0,0)}$ is continuous, there exists an optimal solution which satisfies the 
Kuhn-Tucker optimality condition
\begin{equation} \label{control equation of log utility}
 \begin{cases}
    \mu^S-r-(\sigma^S)^2\pi^S - \rho\sigma^S\sigma^P\pi^P - \frac{h_{(0,0)}^S(s,p)}{1-\pi^S-L^P\pi^P} - \frac{L^S h_{(0,0)}^P(s,p)}{1-L^S\pi^S-\pi^P} + \mu_1-\mu_2 = 0 \\
   \mu^P-r-(\sigma^P)^2\pi^P - \rho\sigma^S\sigma^P\pi^S - \frac{L^Ph_{(0,0)}^S(s,p)}{1-\pi^S-L^P\pi^P} - \frac{h_{(0,0)}^P(s,p)}{1-L^S\pi^S-\pi^P} + \mu_3-\mu_4 = 0
\end{cases}
\end{equation}
and the complementary slackness condition
\begin{equation}\label{slackness}
 \mu_1(\pi^S-a^S)=0,\; \mu_2(b^S-\pi^S)=0, \;
\mu_3(\pi^P-a^P)=0, \; \mu_4(b^P-\pi^P)=0,
\end{equation}
where $\mu_i\geq 0$, $i=1,\ldots,4$, are Lagrange multipliers. Since $\pi^S$ can only take value either in the interior of  interval $[a^S,b^S]$ or one of two endpoints, the same applies to $\pi^P$, we have nine possible combinations. 

If both $\pi^S$ and $\pi^P$ are interior points, then $\mu_i=0$ for $i=1,\ldots,4$ from (\ref{slackness}). Assuming that there exists a unique solution $\left((\pi^S)^{*}_{(0,0)},(\pi^P)^{*}_{(0,0)}\right)$ of (\ref{control equation of log utility}) such that $(\pi^S)^{*}_{(0,0)}\in (a^S, b^S)$ and $(\pi^P)^{*}_{(0,0)}\in (a^P, b^P)$, then $\left((\pi^S)^{*}_{(0,0)},(\pi^P)^{*}_{(0,0)}\right)$ is the optimal control. We can discuss other cases one by one. For example, if $(\pi^S)^{*}_{(0,0)}=a^S$ and $(\pi^P)^{*}_{(0,0)}\in (a^P, b^P)$, then $\mu_2=\mu_3=\mu_4=0$ from (\ref{slackness}) and $(\pi^P)^{*}_{(0,0)}$ and $\mu_1$ are solutions of equation (\ref{control equation of log utility}). If solutions do not satisfy $(\pi^P)^{*}_{(0,0)}\in (a^P, b^P)$ and $\mu_1\geq0$, then this case is impossible. 

\begin{remark} \label{only one stock left optimal control}
Applying Kuhn-Tucker optimality condition to $G_{(1,0)}(p,\pi)$, we get the explicit optimal control for $z=(1,0)$ as
$$
(\pi^S)^*_{(1,0)} = 0, \quad (\pi^P)^*_{(1,0)} = \frac{\mu^P-r+(\sigma^P)^2 - \sqrt{(\mu^P-r-(\sigma^P)^2)^2+4(\sigma^P)^2h_{(1,0)}^P(p)}}{2(\sigma^P)^2},
$$
provided $(\pi^P)^*_{(1,0)}\in (a^P,b^P)$, otherwise, $(\pi^P)^*_{(1,0)}$ equals $a^P$ or $b^P$. 
\end{remark}

\begin{remark}\label{CF(2007)comments}
Capponi and Frei (2017) derive explicit optimal trading strategies for log utility investors when there are $N$ stocks and $N$ CDSs for these stocks. Applying Ito's formula to log wealth process and taking expectation, they get
\begin{equation} \label{capponi and frei equation}
\mathbb{E}[\ln X_T] = \ln x + \int_0^T \mathbb{E}[\alpha_t] dt,
\end{equation}
where $\alpha_t:=r + f(\bar{x})+\sum_{n\in I_z} h_ng_n(y_n)$ and $\bar{x}$ is a vector of dimension $N_z$ such that each component is a linear combination of $N_z$ controls $\pi$ into stocks and $N_z$ controls $\psi$ into CDSs, and $y_n$, $n\in I_z$, are similarly defined.
To maximize $\alpha_t$ over controls $\pi$ and $\psi$, Capponi and Frei (2017)  use a clever trick of maximizing $f(\bar{x})$ and $g_n(y_n)$ separately and derive a linear equation system with $2N_z$ equations and $2N_z$ variables in $\pi$ and $\psi$. The explicit optimal controls come from solving the equation system, see equation (B.3) in the E-companion paper of Capponi and Frei (2017).

The success of finding the explicit optimal control in Capponi and Frei (2017) crucially relies on the existence of equal number of CDSs in the model. When there is no CDS in the portfolio as in our case, maximizing $f(\bar{x})$ and $g_n(y_n)$ separately would result in  an incompatible system of $2N_z$ equations with  $N_z$ variables.  It is therefore impossible to get the closed-form optimal control for log utility investors in our looping contagion model by applying Capponi and Frei's technique. In fact, applying Ito's formula  to log wealth process and taking  expectation in our model, we get 
$$ \mathbb{E}[\ln X_T] =\ln x +  \int_0^T\mathbb{E}[\widetilde{\alpha}_t] dt,
$$
where $\widetilde{\alpha}_t:=r+\pi_t^T D_t \theta-\frac{1}{2}\pi_t^TD_t\Sigma D_t \pi_t + \sum_{j\in I_z} h_z^j(s)\ln\left(1-\sum_{i\in I_z} L_{ij}\pi_{t-}^i\right)$. Taking derivatives of $\widetilde{\alpha}_t$ with respect to $\pi$ would lead to the same  equation system as  that in (\ref{control equation of log utility}).
\end{remark}

\subsection{Performance comparison of  state-dependent and constant intensities}
We now do some statistical analysis. The data used are the same as the benchmark case and:
$$T=1, S_0=100, P_0=100, x_0=100.$$
Assume the intensity function  $h$ is given by

\begin{equation} \label{h}
h(x,y) = \min \left\{ \max \left \{ h_0 \left(k_1x + k_2y\right)^{-\alpha}, h_m \right \}, h_M \right \}
\end{equation}
with minimum intensity $h_m=0.05$, maximum intensity $h_M=1.0$, and  parameter~$\alpha=1$. The default intensity functions with respect to each stock and default state are given by 
$$
h^S_{(0,0)}(s,p) = h(s,p), \quad h^P_{(0,0)}(s,p) = h(p,s), \quad h^S_{(0,1)}(s) = h(s,0), \quad h^P_{(1,0)}(p) = h(p,0).
$$
Note that $h_0$ controls the initial intensity and weights $k_1, k_2$ control the sensitivity of intensity $h$ to stock prices $s$ and $p$.   We set $h_0=10.0$ such that the initial intensity is 0.1 and $k_1=0.7, k_2=0.3$ which means the default intensity of one stock is slightly more sensitive to its own stock price. Moreover, the intensity of one stock jumps up when the other stock defaults, which captures the virtue of interacting default intensity model, see Bo and Capponi (2013). 

\begin{figure}[!htb]
\hfill \minipage{0.31\textwidth}
  \includegraphics[width=\linewidth]{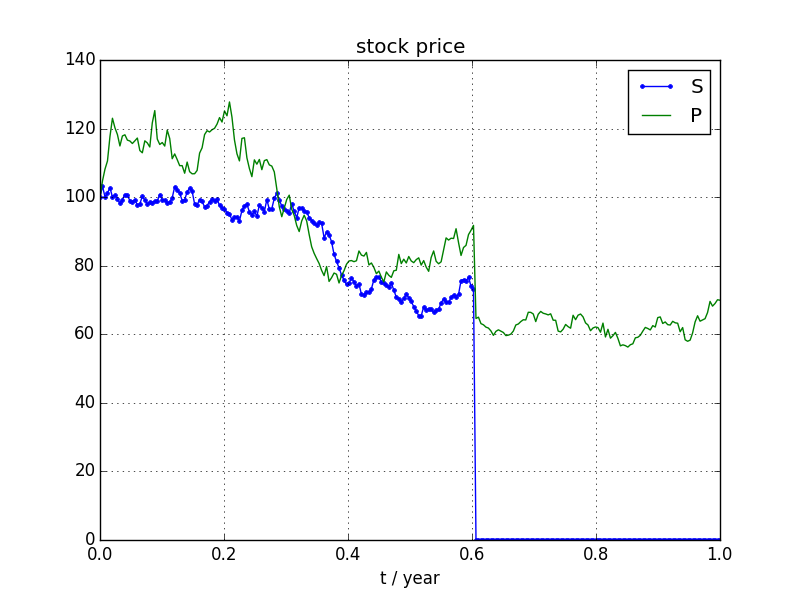}
\endminipage
\minipage{0.31\textwidth}
  \includegraphics[width=\linewidth]{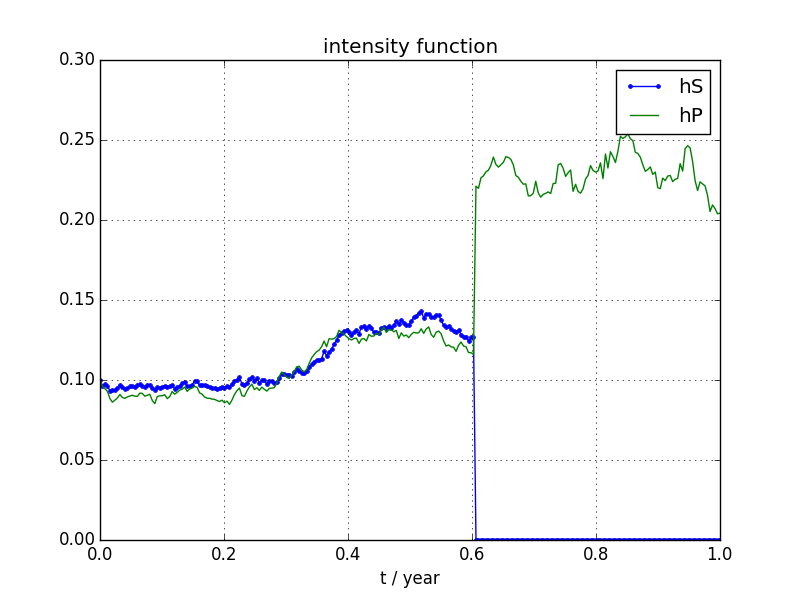}
\endminipage\hfill
\minipage{0.31\textwidth}
\includegraphics[width=\linewidth]{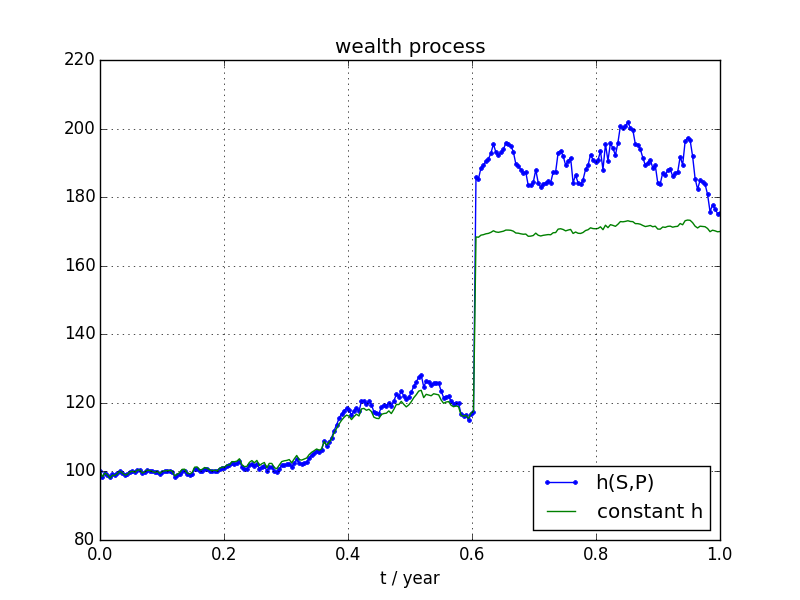}
\endminipage\hfill
\caption{Sample paths of stock price, default intensity, and wealth}
\label{fig:simulated stock price paths and corresponding intensity paths}
\end{figure}

Figure \ref{fig:simulated stock price paths and corresponding intensity paths} shows  sample paths of stock prices, default intensities, and optimal wealth with two different trading strategies. The left panel shows stock price paths of $S$ and $P$. In this scenario, only stock $S$ defaults. At time of default, stock price $S$ drops to zero and stock price $P$ jumps down then continues. The middle  panel shows the default intensity processes $h^S_z(S_t,P_t)$ and $h^P_z(S_t,P_t)$, which are functions of stock prices $S_t, P_t$. The intensity of stock $S$ becomes zero after default, while the default intensity $h^P_{(0,0)}(S_t,P_t)$ jumps up to $h^P_{(1,0)}(P_t)$.  The right panel shows the sample wealth paths when optimal control strategies used are based on $S,P$-dependant intensities and constant intensities (value equal to $0.1$). Both the wealth paths with $S,P$-dependant intensities and constant intensities jump up when default occurs and then two wealth paths move in the same pattern. Compared with the constant intensities, the $S,P$-dependant wealth path jumps more. This is not surprising as at time of default, strategies with intensity $h(S_t,P_t)$ short sells more stocks $S$ and $P$ than strategies with constant intensity, which means gain is more, see Figure \ref{fig: optimal trading strategies and final wealth distribution}. Of course, this is due to the fact that at time of default the default intensities of $S$ and $P$ are both above 0.1. The opposite phenomenon happens when the intensity $h^P(S_t,P_t)$ at time of default is larger than constant intensity 0.1.

\begin{figure}[H]
\hfill \minipage{0.31\textwidth}
  \includegraphics[width=\linewidth]{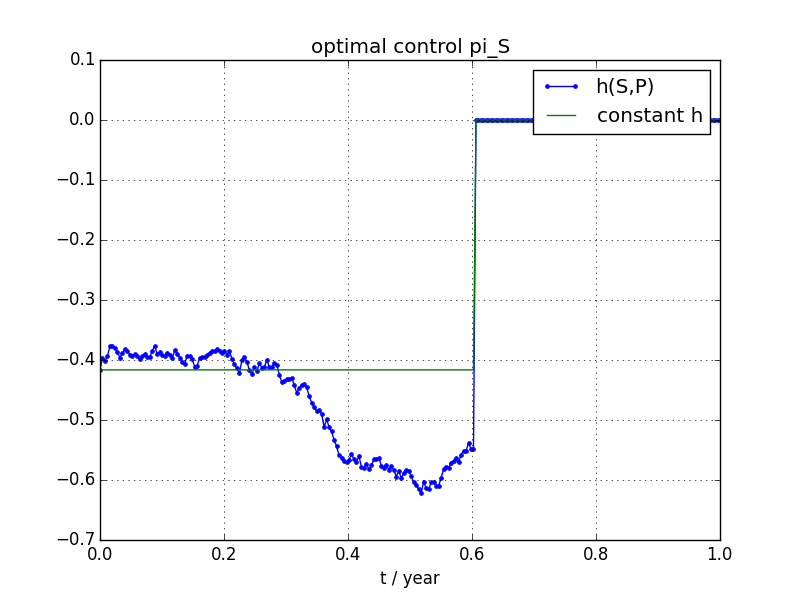}
\endminipage
\minipage{0.31\textwidth}
  \includegraphics[width=\linewidth]{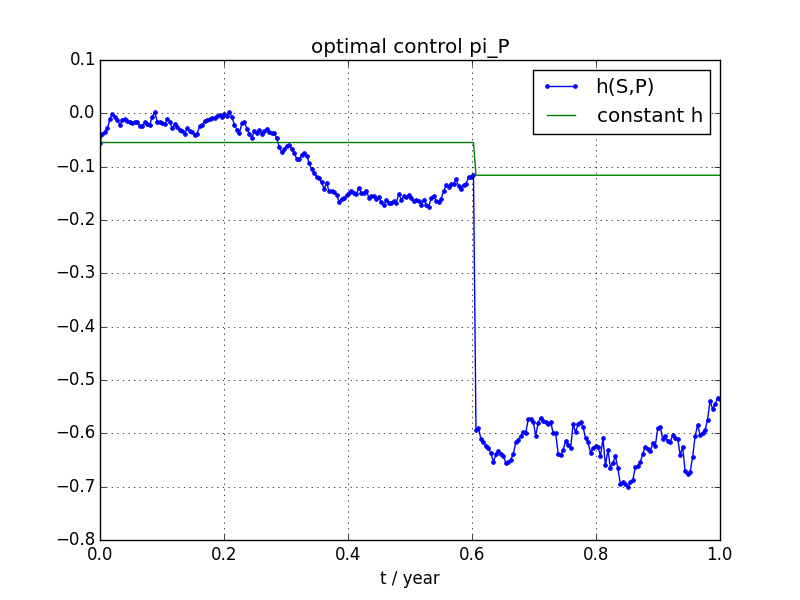}
\endminipage
\minipage{0.31\textwidth}
  \includegraphics[width=\linewidth]{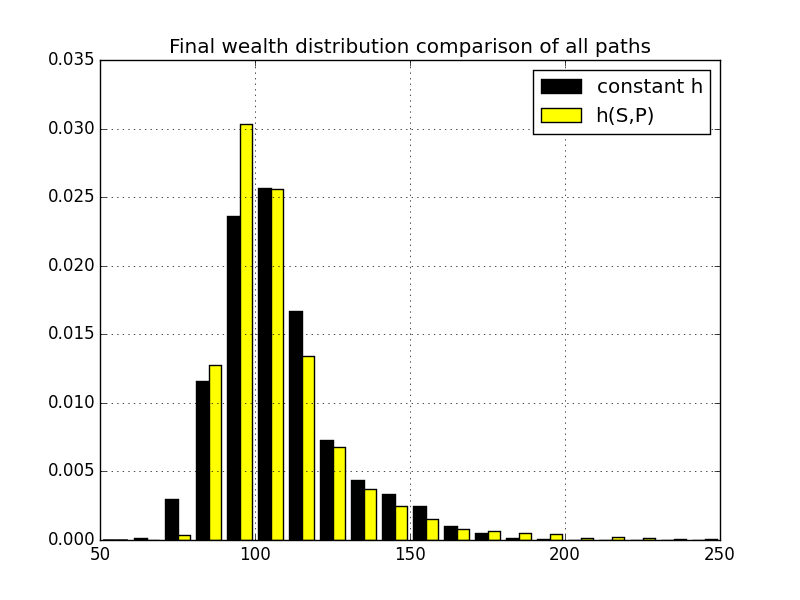}
\endminipage
\hfill 
\caption{Optimal controls and terminal wealth distribution.}
\label{fig: optimal trading strategies and final wealth distribution}
\end{figure}

Figure \ref{fig: optimal trading strategies and final wealth distribution} shows  optimal controls $\pi^S$ and $\pi^P$ associated with the stock paths in Figure \ref{fig:simulated stock price paths and corresponding intensity paths} and the statistical distributions of the wealth at time $T$.  
 The left panel and mid panel are  proportions of wealth invested in stocks $S$ and $P$, respectively. It is clear that as default intensity increases,  investments in stocks  $S$ and $P$ both decrease and investment in savings account $B$ increases, which is intuitively expected as if the default probability of one stock increases, then one would reduce the holdings of both stocks $S$ and $P$ to reduce the risk of loss in case the default of system indeed occurs. In this scenario, both the optimal investment strategies to stock $S$ and $P$ are short-selling, which is a combination effect of parameters chosen and default in the place. 
    We simulate 10000 paths of both stock prices $S$ and $P$, using $S,P$-dependent default intensity $h(S_t,P_t)$. Among all these paths, about 1/5 (precisely 1752 paths) contain defaults of either $S$ or $P$.  The terminal wealth is generated by two strategies: one is  optimal strategy based on the full information of $h(S_t,P_t)$, the other is optimal strategy based on constant intensity 0.1. 
The right panel shows the histograms of terminal wealth of these two strategies. 
It is clear that their distributions are similar but with some slight differences at tail parts, that is, probability of over-performance is higher and probability of under-performance is lower with $S,P$-dependant intensities.
These histograms seem to indicate, for log utility,  the overall performance of $S,P$-dependent optimal strategies and constant strategies are similar, while the $S,P$-dependent optimal strategies perform better in extreme scenarios. 
                      
\begin{table}[!htbp]   
\begin{center}
  \begin{tabular}{| l | c | c | c| c|}
  \hline 
     & mean & std dev &  2.3\% quantile & 97.7\% quantile  \\ \hline
       All samples $+$ $h(S,P)$ & 107.78 & 22.60 & 84.08 & 171.82  \\ \hline 
    All samples $+$ constant $h$ & 107.59 & 19.27 & 78.60 & 157.80  \\ \hline
    Default $+$ $h(S,P)$ & 134.81 & 36.04 & 83.80 & 225.20 \\ \hline 
    Default $+$ constant $h$ & 134.68 & 21.50 & 95.29 & 178.04  \\ \hline 
    No-default $+$ $h(S,P)$ & 102.17 & 12.82 & 84.11 & 134.68 \\ \hline 
    No-default $+$ constant $h$ & 101.96 & 12.99 & 78.13 & 130.48  \\ \hline
    \end{tabular}
    \end{center} 
    \caption{Sample means, standard deviations, and quantile values.} 
    \label{table1}
\end{table} 

Table \ref{table1}  contains sample means, sample standard deviations, and quantile values at low end (2.3\%) and high end (97.7\%) for both $S,P$-dependant intensity and constant intensity $0.1$. It is clear that the overall sample mean with $S,P$-dependent optimal strategies is (slightly) higher than constant intensity optimal strategies, which is expected as the former is the genuine optimal control,  however, the sample standard deviation with $S,P$-dependent optimal strategies is also higher, which implies the $S,P$-dependent optimal strategies can be volatile and risky, while the constant optimal strategies are more conservative. However, if we check the quantiles of the distribution (which is a different risk measure), we find that the $S,P$-dependent optimal strategies overall generate both higher 2.3\% quantile (less loss) and higher 97.7\% quantile (more gain), which implies the $S,P$-dependent optimal strategies outperform the constant strategies in the extreme scenarios. Note that the outperformance in upper quantile comes from more short selling (anticipating the default when stock price is very low). 

\begin{remark} \label{other parameters}
By far the conclusion drawn relies on the benchmark parameter values, in which case the optimal controls for both stocks are short selling in most scenarios. We repeat the same comparison tests on two other parameter sets (if not specified, the parameter value is the same as benchmark case).
\begin{itemize}
	\item Parameter set 1. $\sigma^S=0.2, \sigma^P=0.3, L^S=0.1, L^P = 0.2, \rho=0.4, h_0=5.0, h_m=0.01.$
	\item Parameter set 2. $r=0.01, \mu^S=0.15, \mu^P=0.2, L^S=0.05, L^P=0.1, \rho = 0.7, h_0=5.0, h_m=0.01.$
\end{itemize}
 In most scenarios, the strategies of parameter set 1 is short selling stock $S$ and longing stock $P$ and the strategies of parameter set 2 are longing both stocks $S,P$. 
The overall performance with $S,P$-dependent optimal strategies is very similar to that with constant intensity optimal strategies. The overall sample mean with $S,P$-dependent optimal strategies is (slightly) higher than constant intensity optimal strategies, and the sample standard deviation with $S,P$-dependent optimal strategies is also higher, which implies the difference in the tail distribution and that $S,P$-dependant optimal strategies can be more volatile.
\end{remark}

\subsection{Robust  tests of model parameters}
Assume intensity function $h$ is given by (\ref{h}) and stock prices $S$ and $P$ are generated  
based on that. It may be difficult to calibrate parameters  accurately even one knows the exact form of the intensity function. We do some robust tests for parameters $k_1, k_2, \alpha, h_m, h_M, h_0$, that is,  
we compare the optimal performances of two investors, one uses benchmark parameter values and the other incorrect estimated values. We change one parameter  only in each test while keep all other parameters fixed at benchmark values.

\begin{table}[!htbp] 
\begin{center}
\begin{tabular}{| l | c | c | c| c|}
  \hline 
     & mean & std dev &  2.3\% quantile & 97.7\% quantile  \\ \hline
    benchmark  & 107.78 & 22.60 & 84.08 & 171.82 \\ \hline
       $k_1=0.5,k_2=0.5$ & 107.69 (-0.09\%) & 22.85 (1.10\%) & 83.00 (-1.29\%) & 173.82 (1.16\%) \\ 
       $k_1=0.3,k_2=0.7$ & 107.69 (-0.09\%) & 23.15 (2.44\%) & 80.95 (-3.73\%) & 174.67(1.66\%) \\ \hline  
       $\alpha=0.8$ & 111.59 (3.53\%) & 53.15 (135.20\%) & 53.18 (-36.75\%) & 262.58 (52.82\%) \\ 
        $\alpha=1.2$ & 105.17 (-2.42\%) & 9.43 (-58.26\%) & 84.52 (0.53\%) & 122.93 (-28.45\%) \\ \hline 
     $h_m=0.01, h_M=1.5$ & 107.77 (-0.01\%) & 22.60 (0.01\%) & 84.08 (0.00\%) & 171.82 (0.00\%) \\ 
    $h_m=0.07, h_M=0.5$ & 107.76 (-0.02\%) & 22.62 (0.07\%) & 83.91 (-0.20\%) & 171.82 (0.00\%) \\ \hline 
   $h_0=5$ & 105.36 (-2.25\%) & 9.52 (-57.86\%) & 85.24 (1.39\%) & 124.36 (-27.62\%) \\
 $h_0=15$ & 109.76 (1.84\%) & 37.64 (66.55\%) & 70.59 (-16.04\%) & 221.49 (28.91\%) \\ 
    \hline
    \end{tabular}
    \end{center} 
  \caption{Robust test of intensity parameters} 
  \label{table4}
\end{table}

Table \ref{table4} shows that sample means are essentially the same over a broad range of model parameters. The main difference is sample standard deviations. Percentage changes over the benchmark values are listed in parentheses. The performance of state-dependent intensity strategies is robust for some parameters, including weight $k_1, k_2$, minimum intensity level $h_m$ and maximum intensity level $h_M$. Changes of these parameters do  not greatly change sample standard deviations and quantile values at low and high ends.   On the other hand, it seems important to have correct estimations of parameters $\alpha$ and $h_0$  
to avoid large changes of the standard deviation. Those parameters have strong impact on the estimated intensity levels. For example, if one overestimates the initial default intensity ($h_0=15$ instead of correct value $h_0=10$) then the sample standard deviation is greatly increased with large loss at low end quantile value. 

Next we do some robust tests to see the impact of changes of model parameters on the distribution of optimal terminal wealth, including drift $\mu$, volatility $\sigma$, correlation $\rho$, and percentage loss $L^S$. 
 We change drift and volatility parameters by 20\% of their benchmark values and correlation and percentage loss parameters by some big  deviations. 

\begin{table}[!htbp] 
\begin{center}
  \begin{tabular}{| l | c | c | c| c|}
  \hline 
     &  mean & std dev &  2.3\% quantile & 97.7\% quantile  \\ \hline
   benchmark  & 107.78 & 22.60 & 84.08 & 171.82  \\ \hline
      $\mu^S=0.12$ & 107.05 (-0.68\%) & 17.54 (-22.39\%) & 91.98 (9.39\%) & 158.13 (-7.97\%)  \\
  $\mu^S=0.08$  & 108.55 (0.71\%) & 28.35 (25.44\%) & 75.99 (-9.62\%) & 187.88 (9.35\%) \\  
          $ \mu^P=0.18$ & 107.53 (-0.23\%) & 21.78 (-3.61\%) & 80.93 (-3.75\%) & 165.05 (-3.94\%) \\ 
   $ \mu^P=0.12$  & 108.08 (0.28\%) & 25.58 (13.17\%) & 83.68 (-0.48\%) & 182.54 (6.24\%) \\  \hline
      $\sigma^S=0.36$  & 107.28 (-0.46\%) & 18.87 (-16.49\%) & 88.04 (4.71\%) & 161.89 (-5.78\%)  \\
    $\sigma^S=0.24$  & 108.47 (0.64\%) & 27.88 (23.37\%) & 78.37 (-6.79\%) & 188.35 (9.62\%) \\ 
      $ \sigma^P=0.48$  & 107.74 (-0.04\%) & 21.94 (-2.92\%) & 84.12 (0.04\%) & 169.01 (-1.63\%) \\ 
 $\sigma^P=0.32$ & 107.86 (0.07\%) & 23.56 (4.24\%) & 83.84 (-0.29\%) & 175.82 (2.33\%) \\ \hline
     $\rho=-0.3$ & 108.21 (0.40\%) & 25.50 (12.83\%) & 83.45 (-0.75\%) & 183.56 (6.83\%) \\
   $\rho=0.3$  & 107.57 (-0.20\%) & 21.83 (-3.39\%) & 82.74 (-1.59\%) & 167.26 (-2.65\%)  \\ \hline  
   $L^S=0.1$  & 107.49 (-0.26\%) & 20.49 (-9.33\%) & 87.32 (3.85\%) & 166.49 (-3.10\%) \\
   $L^S=0.4$  & 108.24 (0.43\%) & 26.45 (17.05\%) & 77.52 (-7.80\%) & 181.95 (5.90\%) \\
   $L^P=0.15$  & 107.70 (-0.07\%) & 21.94 (-2.91\%) & 83.26 (-0.97\%) & 169.43 (-1.39\%) \\
   $L^P=0.6$  & 107.88 (-0.10\%) & 23.90 (5.74\%) & 84.69 (0.73\%)& 177.53 (3.32\%) \\ 
   \hline
       \end{tabular}
    \end{center} 
    \caption{Robust test of model parameters} 
    \label{table5}
\end{table}

Table \ref{table5} lists statistical results of distributional sensitivity to changes of parameters. It is clear that sample means are essentially the same for all parameters, but sample standard deviations are sensitive to changes of drift, volatility, correlation and percentage loss, which would significantly affect overall distributions of optimal terminal wealth. This requires one to have good estimations of these parameters to have correct distributions. It is well known that it is easy to estimate volatility but difficult to estimate drift (see Rogers (2013)) and information of percentage loss is rarely available. Since optimal trading strategies and optimal wealth distributions are greatly influenced by these parameters which are difficult to be correctly estimated, one needs to be cautious                                                                                                                                                                                                                                                                                                                                                                                                                                                                                                                                                                                                                                                                                                                                                                                                                                                                                                                                                                                                                                                                                                                                                                                                                                                                                                                                            in using state-dependent intensity to model and solve optimal investment problems. Using sub-optimal but conservative and robust trading strategies, instead of optimal ones based on unobservable parameters and intensities, might be more sensible  and less risky.

\subsection{Performance comparison of different initial stock prices}
Table \ref{table1} shows the overall distributions of the terminal wealth are similar whether one uses the intensity $h^i_z(s, p)$ or constant intensity 0.1 as approximation. This is possibly due to the fact that the initial price of $S$ and $P$ are both 100, which results in the inital intensity $h^i_z(s,p)$ being equal to the constant intensity.  The value 0.1 comes from the calibration which relies only on the historical data, while $h(s,p)$ is a forward-looking function which depends on the future stock prices. Table \ref{table1} represents the normal situation where the default probability in calibration window is close to that in investment window. However, if the initial intensity $h^i_z(s,p)$ is vastly different from $0.1$ which comes from the estimation of calibration window (one example is that the calibration window is just before the financial crisis, while the investment starts from the financial crisis period), the distributions of terminal wealth  can be significantly different. We use a numerical example to illustrate this. 

For simplicity, let the intensity function be given by $h^i_z(s,p)=20/(s+p)$. This means the default intensity of $S$ jumps from $20/(S_t+P_t)$ to $20/S_t$ after $P$ defaults. So is the situation when $S$ defaults. Assume that the initial prices of $S$ and $P$ are $s=10, p=10$ respectively, then the initial intensity is $h^S_{(0,0)}(s,p)=h^P_{(0,0)}(s,p)=1$, which makes the stocks ten times more likely to default than the constant intensity $h=0.1$ would have suggested (from the calibration window). This would cause one to take different control strategies (more shortselling when $s=10, p=10$) and would have large impact on the distributions of the terminal wealth as shown in the table below.

\begin{table}[!htbp]   
\begin{center}
  \begin{tabular}{| l  | c | c | c| c|}
  \hline 
     &  mean & std dev &  2.3\% quantile & 97.7\% quantile  \\ \hline
       All samples $\&$ $h(S,P;S_0=10,P_0=10)$ &  179.74 & 63.21 & 58.63 & 316.08  \\ \hline 
    All samples $\&$ $h\equiv 0.1$ & 66.26 & 22.87 & 44.62 & 138.28  \\ \hline
    Default $\&$ $h(S,P;S_0=10,P_0=10)$ & 195.90 & 52.78 & 99.20 & 318.65 \\ \hline 
    Default $\&$ $h\equiv 0.1$ & 57.70 & 7.16 & 44.23 & 73.82  \\ \hline 
    No-default $\&$ $h(S,P;S_0=10,P_0=10)$ & 88.18 & 31.46 & 50.36 & 172.72 \\ \hline 
    No-default $\&$ $h\equiv 0.1$ & 114.78 & 20.69 & 78.99 & 157.11  \\ \hline
    \end{tabular}
    \end{center} 
    \caption{Sample statistics. Data: $x=100$, $S_0=10$, $P_0=10$, $T=1$} 
    \label{table6}
\end{table}

Table \ref{table6} shows the statistics of the terminal wealth with 10000 simulation scenarios which produces 8542 default scenarios, a reflection of the high initial default intensity $h^i_z(S_0,P_0)=1$.   
When stock prices are small, defaultable stocks are very likely to default. With $S,P$ dependent intensity, the optimal controls are to short sell more stocks, which results in a much larger mean (195.90) than the mean (57.70) with constant intensity $h\equiv 0.1$ if stock $S$ or $P$ indeed defaults (anticipated). However, if stock does not default (non-anticipated), then the opposite outcomes appear. This numerical test shows $S,P$-dependant control strategies may outperform or under-perform $S,P$ independent control strategies, depending on the anticipated market event (default of stock) occurring or not. 

\begin{remark}
We repeat the same tests on the other two parameter sets defined in Remark \ref{other parameters}. Numerical results display the similar patterns as those in Table \ref{table6}, so the conclusion drawn in this section is robust.

\end{remark}

\subsection{Numerical method for  power utility} \label{section_power}
For power utility  $U(x) = (1/\gamma)x^\gamma$, $0<\gamma<1$, the post-default case  is well known with the optimal control  $\pi^S =(\mu^S-r)/((\sigma^S)^2(1-\gamma))$ (and $\pi^P = 0$) and the  post-default value function 
$v_1(t,x) =(1/\gamma) x^\gamma g_1(t)$, where
$$ g_1(t)= \exp{\left(\left( r\gamma+\frac{\gamma}{2(1-\gamma)} \left(\frac{\mu^S-r}{\sigma^S}\right)^2\right)(T-t)\right)}.$$
We conjecture that  the pre-default value function  takes the form
\begin{equation} \label{value function of power utility}
w(t,x,s,p) = \frac{x^\gamma}{\gamma} f(t,s,p).
\end{equation}
Substituting (\ref{value function of power utility}) into (\ref{General pre-default HJB equation}), we get  a semilinear PDE for $f$:
\begin{equation} \label{Power utility HJB equation}
-\frac{\partial f}{\partial t} -\frac{1}{2}\Tr\left(\sigma\sigma^T(s,p)\mathcal{D}^2f\right)
-\sup_{\pi\in A}\left\{ b^T(s,p,\pi)\mathcal{D}f - \beta(s,p,\pi)f + g(t,s,p,\pi)\right\}=0,
\end{equation}
with terminal condition $f(T,s,p) = 1$, where 
$$
b(s,p,\pi): =  \begin{pmatrix}  \left(\mu^S+\gamma m^T\pi\sigma^S\right)s\\   \left(\mu^P+\gamma n^T\pi\sigma^P\right)p \end{pmatrix}, \;
\mathcal{D}f: =  \begin{pmatrix}  f_s\\   f_p \end{pmatrix}, \;
\sigma(s,p):=
 \begin{pmatrix}
  \sigma^Ss & 0 \\
 \rho\sigma^Pp& \sqrt{1-\rho^2}\sigma^Pp
   \end{pmatrix},\;
\mathcal{D}^2f:=
 \begin{pmatrix}
  f_{ss} & f_{sp} \\
  f_{sp} & f_{pp}
 \end{pmatrix},  
$$
and 
$$
\beta(s,p,\pi):=-r\gamma+h(s,p)-\gamma\left(\theta^T\pi+\frac{1}{2}(\gamma-1)\pi^T\Sigma\pi\right),
$$
$$
g(t,s,p,\pi):=h(s,p)g_1(t)(1-L^T\pi)^\gamma.
$$

Equation (\ref{Power utility HJB equation}) is a nonlinear PDE with two state variables and it is highly unlikely, if not impossible, to find a closed form solution $f$. However,
by Pham (2009) (Remark 3.4.2), equation  (\ref{Power utility HJB equation}) is 
the HJB equation for the value function $v$ of  the following optimal control problem:
\begin{equation} \label{power utility value function}
v(t,y) = \sup_{\pi\in\mathcal{A}} \mathbb{E}\left[\int_t^T \Gamma(t,u)g(u,Y_u,\pi_u)du + \Gamma(t,T)\bigg| Y_t=y\right],
\end{equation}
where $Y_u:=(S_u,P_u)^T$, $t\leq u\leq T$, is a controlled Markov state process satisfying the following SDE: 
\begin{equation} \label{dynamic of power wealth}
dY_u = b(Y_u,\pi_u)du + \sigma(Y_u)dW_u,\ t\leq u\leq T,
\end{equation}
with the initial condition $Y_t = y:=(s,p)^T$, 
$W$ is a 2-dimensional standard Brownian motion and 
$\Gamma(t,u):=\exp\left\{ -\int_t^u \beta(Y_l,\pi_l) dl\right\}$ is a discount factor. 

By our theoretical result, we claim that the value function $v(t,y)$ is the unique viscosity solution of the HJB equation (\ref{Power utility HJB equation}). Moreover, if the HJB equation (\ref{Power utility HJB equation}) has a classical solution, then  it is the value function  $v(t,y)$. In other words, we may find the solution $f(t,s,p)$ of equation (\ref{Power utility HJB equation}) by solving a stochastic optimal control problem (\ref{power utility value function}). Since the diffusion coefficient of SDE (\ref{dynamic of power wealth}) does not contain control variable $\pi$, we may use the numerical method of Kushner and Dupuis (2001) to find the optimal  value function in (\ref{power utility value function}), which would give us a numerical approximation to the solution  $f(t,s,p)$ of equation (\ref{Power utility HJB equation}). Next we give some details.

According to Kushner and Dupius (2001),   the process $Y$ can be approximated by a Markov chain process, which transits a point $Y_t=(s,p)$ at time $t$ to one of nine points $Y_{t+\Delta t}$ may take at time $t+\Delta t$, that is, $(s,p), (s\pm \delta,p), (s,p\pm \delta), (s+\delta,p\pm \delta), (s-\delta,p\pm \delta)$,  with the following transition probabilities:
\begin{eqnarray*}
a^{\delta,\Delta t}\left((s,p),(s,p)\mid \pi\right) &:= & 1 - \frac{\Delta t}{\delta}(|b_1|+|b_2|) - \frac{\Delta t}{\delta^2}\left((\sigma^Ss)^2+(\sigma^Pp)^2 - |\rho|\sigma^S\sigma^Psp\right) \\
a^{\delta,\Delta t}\left((s,p),(s\pm \delta,p)\mid \pi\right) &:= & \frac{\Delta t}{\delta}b_1^{\pm} + \frac{\Delta t}{2\delta^2}(\sigma^Ss)^2 - \frac{\Delta t}{2\delta^2}|\rho|\sigma^S\sigma^Psp \\
a^{\delta,\Delta t}\left((s,p),(s,p\pm \delta)\mid \pi\right) &:= & \frac{\Delta t}{\delta}b_2^{\pm} + \frac{\Delta t}{2\delta^2}(\sigma^Pp)^2 - \frac{\Delta t}{2\delta^2}|\rho|\sigma^S\sigma^Psp \\
a^{\delta,\Delta t}\left((s,p),(s+\delta,p\pm \delta)\mid \pi\right) &:= & \frac{\Delta t}{2\delta^2}\rho^{\pm}\sigma^S\sigma^Psp \\
a^{\delta,\Delta t}\left((s,p),(s-\delta,p\mp\delta)\mid \pi\right) &:= & \frac{\Delta t}{2\delta^2}\rho^{\pm}\sigma^S\sigma^Psp, 
\end{eqnarray*}
where $\delta$ is the step size of space,  $\Delta t:=(T-t)/N$ is the step size of time with $N\geq 1$ an integer,  $b_1:=(\mu^S+\gamma m^T\pi\sigma^S)s$, $b_2:=(\mu^P+\gamma n^T\pi\sigma^P)p$ and $x^+:=\max\{x,0\}$, $x^-:=\max\{-x,0\}$.

The numerical scheme is based on the following discretized dynamic programming principle:
\begin{eqnarray*}
&&v(k\Delta t, S_{k\Delta t}, P_{k\Delta t}) \nonumber \\
&\approx& {}  \sup_{\pi_k\in A} \left(g(k\Delta t,S_{k\Delta t},P_{k\Delta t},\pi_k) \Delta t + \exp\left\{ -\beta(S_k^N,P_k^N,\pi_k) \Delta t \right\}
\mathbb{E}\left[v\left((k+1)\Delta t,S_{k+1}^N,P_{k+1}^N\right)\right] \right) 
\end{eqnarray*}
for $k=N-1,\ldots, 1,0$, 
where $\pi_k$ is the piece-wise constant control and the expectaton is computed with the help of the above Markov chain transition probabilities. The terminal condition is given by $v(N\Delta t, S_{N\Delta t},P_{N\Delta t}) = 1$. 

We compare the passive investment and active investment under the power utility setting. Most  parameter values used in power utility case are the same as log utility benchmark case, except  the step size of space $\delta=5$, the step size of time $\Delta t=0.1$, and  the set of control parameters are $a^S=a^P=-1.0, b^S=b^P=1.0$. Table \ref{table9} lists the numerical results with mean, variance, and quantile values at lower and upper ends. It is clear the performance is similar to that of the log utility as one would expect. We have also done other tests defined in log utility scope and drawn the similar conclusions for the power utility investors.
   
\begin{table}[!htbp]   
\begin{center}
  \begin{tabular}{| l | c | c | c| c|}
  \hline 
     & mean & std dev &  2.3\% quantile & 97.7\% quantile  \\ \hline
       All samples $+$ $h(S,P)$ & 106.38 & 17.45 & 77.30 & 148.20  \\ \hline 
    All samples $+$ constant $h$ & 105.99 & 14.90 & 79.92 & 139.08  \\ \hline
    Default $+$ $h(S,P)$ & 106.70 & 21.00 & 71.67 & 158.63 \\ \hline 
    Default $+$ constant $h$ & 106.18 & 15.78 & 78.50 & 140.40  \\ \hline 
    No-default $+$ $h(S,P)$ & 106.35 & 17.06 & 77.83 & 146.52 \\ \hline 
    No-default $+$ constant $h$ & 105.97 & 14.80 & 80.03 & 138.94  \\ \hline
    \end{tabular}
    \end{center} 
    \caption{Sample means, standard deviations, and quantile values.} 
    \label{table9}
\end{table} 

 There is a backward stochastic differential equation (BSDE) representation of the 
 solution  $f(t,s,p)$ of equation (\ref{Power utility HJB equation}). So in theory one may find $f$ if one can solve a highly  nonlinear BSDE, which is  not pursued in this paper, see Cheridito et al. (2007) for details.
 
\section{Proofs 
}\label{proofs}

\subsection{Proof of Theorem~\ref{continuity property of value function}}\label{2.5}
\begin{proof}
We prove the theorem in four steps: 1) $v_0$ is continuous in $x$, uniformly in $t,s,p$, 2)  $v_0$ is continuous in $s$, uniformly in $t,p$, 3)  $v_0$ is continuous in $p$, uniformly in $t,s$ and 4) $v_0$ is continuous in $t$. Combining  these four steps gives the continuity of $v_0$ in $t,x,s,p$.

\textbf{Step 1.} For any $x_1$, $x_2\in [0,\infty)$ and $t,s,p\in [0,T]\times(0,\infty)^2$, using Assumption \ref{utility assumption}, we have 

\begin{equation*}
\begin{split}
|v_0(t,x_1,s,p)-v_0(t,x_2,s,p)| & = \left|\sup_{\pi\in \mathcal{A}} \mathbb{E}\left[U(X_T^{t,x_1,s,p,\pi})\right] - \sup_{\pi\in \mathcal{A}} \mathbb{E}\left[U(X_T^{t,x_2,s,p,\pi})\right]\right| \\
& \leq \sup_{\pi\in \mathcal{A}} \mathbb{E}\left[\left|U(X_T^{t,x_1,s,p,\pi}) - U(X_T^{t,x_2,s,p,\pi})\right|\right] \\
& \leq K \sup_{\pi\in \mathcal{A}} \mathbb{E}\left[\left|X_T^{t,x_1,s,p,\pi} - X_T^{t,x_2,s,p,\pi}\right|^\gamma\right] \\
& \leq K |x_1-x_2|^\gamma.
\end{split}
\end{equation*}
by virtue of (\ref{strong solution property}). Therefore, $v_0$ is continuous in $x$, uniformly in $t,s,p$.

\textbf{Step 2.} 
Fix $0<s_1<s_2<\infty$ and $t,x,p\in [0,T]\times [0,\infty)\times (0,\infty)$. Denote by
$S^i$ the stock price that starts from $s_i$,  $i=1,2$, and $h^i$ and $\tau_i$ the corresponding default intensity and   default time of stock $P$, respectively.
By our model setting, $\tau_i$  can be represented by 
\[
\tau_i:=\inf \left\{s\geq t: \int_t^s h_u^i du \geq \mathcal{X} \right\},
\]
where $\mathcal{X}$ is a standard exponential random variable on the probability space
$(\Omega,\mathcal{G},\mathbb{P})$ and is independent of the filtration $(\mathcal{F}_t)_{t\geq 0}$, which means $\tau_i$ are totally inaccessible stopping times.

Define $\tau_{min}:=\min{\{\tau_1,\tau_2\}}$. It is clear that before $\tau_{min}$, the stock price dynamic is a standard geometric Brownian motion. We have
\[
\mathbb{E}\left[ \left| S_u^1-S_u^2 \right| \mathbb{I}_{\{u< \tau_{min}\}}\right]  \leq K |s_1-s_2|
\]
and
\begin{equation} \label{useful property in continuity}
\mathbb{E}\left[\int_t^{\tau_{min}\wedge \overline{u}} \left | h_u^1-h_u^2 \right | du 
\right]
 \leq K \mathbb{E}\left[\int_t^{\overline{u}} \left| S_u^1-S_u^2 \right| \mathbb{I}_{\{u< \tau_{min}\}}  du \right]  \leq K |s_1-s_2|
\end{equation}
for any $\overline{u}\in [t,T]$. 
 
If there is no jump on interval $[t,T]$, then $\sup_{[t,T]}|H_u^1-H_u^2|=0$ and $X_T^{t,x,s_1,p,\pi}=X_T^{t,x,s_2,p,\pi}$, where $H^i$ is the jump process associated with default time $\tau_i$. If there is at least one jump on interval $[t,T]$, then $\sup_{[t,T]}|H_u^1-H_u^2|=1$ as $\tau_1$ and $\tau_2$ do not jump at the same time. We have the relation
\begin{eqnarray*}
 |X_T^{t,x,s_1,p,\pi}-X_T^{t,x,s_2,p,\pi}| &=& |X_T^{t,x,s_1,p,\pi}-X_T^{t,x,s_2,p,\pi}|\sup_{[t,T]}|H_u^1-H_u^2|\\
&\leq& (|X_T^{t,x,s_1,p,\pi}|+|X_T^{t,x,s_2,p,\pi}|)\sup_{[t,T]}|H_u^1-H_u^2|.
\end{eqnarray*}

Since $\sup_{[t,T]}|H_u^1-H_u^2|$ equals 0 or 1, we have $(\sup_{[t,T]}|H_u^1-H_u^2|)^\alpha=\sup_{[t,T]}|H_u^1-H_u^2|$ for any $\alpha>0$. 
Using $(x+y)^\gamma\leq x^\gamma+y^\gamma$ for $x,y\geq 0$ and $0<\gamma\leq 1$ and  the Cauchy-Schwarz inequality, also noting Remark \ref{Proposition of wealth process}, we have
\begin{eqnarray*}
\mathbb{E}\left[\left|X_T^{t,x,s_1,p,\pi}-X_T^{t,x,s_2,p,\pi}\right|^\gamma\right]
&\leq& \mathbb{E}\left[(|X_T^{t,x,s_1,p,\pi}|^\gamma+|X_T^{t,x,s_2,p,\pi}|^\gamma) \sup_{[t,T]}|H_u^1-H_u^2|\right]\\
&\leq&K\left( \left(\mathbb{E}\left[|X_T^{t,x,s_1,p,\pi}|^{2\gamma}\right]\right)^{1/2}
+ \left(\mathbb{E}\left[|X_T^{t,x,s_2,p,\pi}|^{2\gamma}\right]\right)^{1/2}\right) 
 \left(\mathbb{E}\left[\sup_{[t,T]}|H_u^1-H_u^2|\right]\right)^{1/2}\\
 &\leq& K x^\gamma \left(\mathbb{E}\left[\sup_{[t,T]}|H_u^1-H_u^2|\right]\right)^{1/2}.
\end{eqnarray*}

We therefore have
\[
\begin{split}
|v_0(t,x,s_1,p)-v_0(t,x,s_2,p)| & \leq K \sup_{\pi\in \mathcal{A}} \mathbb{E}\left[|X_T^{t,x,s_1,p,\pi}-X_T^{t,x,s_2,p,\pi}|^\gamma\right] \\
&\leq  K x^\gamma \left(\mathbb{E}\left[\sup_{[t,T]}|H_u^1-H_u^2|\right]\right)^{1/2}.
\end{split}
\]

We can decompose $H^i$ as
$H_u^i = M_u^i + A_u^i$, where $M^i$ 
is a martingale and $A_u^i:= \int_t^{u\wedge \tau_i}h_s^i ds$ is a bounded variation process, see Bielecki and Rutkowski (2003). Applying Doob's sub-martingale inequality, we have 
\begin{eqnarray*}
 \mathbb{E}\left[\sup_{[t,T]}|H_u^1-H_u^2|\right]  
 &=& \mathbb{E}\left[\sup_{[t,T]}|H_u^1-H_u^2|^2\right]  \\
 &\leq&  2 \mathbb{E}\left[\sup_{[t,T]}|M_u^1-M_u^2|^2+ \sup_{[t,T]}|A_u^1-A_u^2|^2\right]\\
 &\leq&  8 \mathbb{E}\left[|M_T^1-M_T^2|^2\right] + 2\mathbb{E}\left[\sup_{[t,T]}|A_u^1-A_u^2|^2\right].
\end{eqnarray*}

Since $h$ is a monotone function of $s$ by Assumption \ref{h assumption}, without loss of generality, we assume $h$ is non-increasing in $s$, then $h^1\geq h^2$ before the first default occurs. By the definition of $\tau_i$, we have $\tau_1\leq \tau_2$ and $H_t^1\geq H_t^2$. Then
\[
\begin{split}
\left|A_u^1-A_u^2\right| 
& = \left|\int_t^{u}h_s^1 ds-\int_t^{u}h_s^2 ds\right| \mathbb{I}_{\{u\leq \tau_1\leq \tau_2 \}} + \left|\int_t^{\tau_1}h_s^1 ds-\int_t^{u}h_s^2 ds\right| \mathbb{I}_{\{\tau_1< u\leq \tau_2 \}} \\
& \ \ \ + \left|\int_t^{\tau_1}h_s^1 ds-\int_t^{\tau_2}h_s^2 ds\right| \mathbb{I}_{\{\tau_1\leq \tau_2<u \}} \\
& = \left(\int_t^{u}h_s^1 ds-\int_t^{u}h_s^2 ds\right) \mathbb{I}_{\{u\leq \tau_1\leq \tau_2 \}} + \left(\mathcal{X}-\int_t^{u}h_s^2 ds\right) \mathbb{I}_{\{\tau_1< u\leq \tau_2 \}} \\
& \ \ \ + \left(\mathcal{X}-\mathcal{X}\right) \mathbb{I}_{\{\tau_1\leq \tau_2<u \}} \\
& = A_u^1-A_u^2
\end{split}
\]
for any $u\in [t,T]$. Therefore,
\[
\sup_{[t,T]}|A_u^1-A_u^2|^2 \leq K \sup_{[t,T]}|A_u^1-A_u^2| = K \sup_{[t,T]}(A_u^1-A_u^2).
\]
Note that $A_u^1-A_u^2$ is non-decreasing before $\tau_1\wedge T$ and non-increasing after $\tau_1\wedge T$, we conclude that
\[
\sup_{[t,T]}(A_u^1-A_u^2) = A_{\tau_1\wedge T}^1-A_{\tau_1\wedge T}^2=\int_t^{\tau_1\wedge T} (h_u^1-h_u^2) du.
\]
By inequality (\ref{useful property in continuity}), we have
\begin{equation} \label{second part}
\mathbb{E}\left[\sup_{[t,T]}|A_u^1-A_u^2|^2\right] 
\leq K\mathbb{E}\left[\int_t^{\tau_1\wedge T} (h_u^1-h_u^2) du\right] 
\leq K|s_1-s_2|.
\end{equation}
Since $H_T^1-H_T^2$ equals 0 or 1, we have
\[
\begin{split}
|M_T^1-M_T^2|^2 & \leq 2|H_T^1-H_T^2|^2 + 2|A_T^1-A_T^2|^2 \\
& \leq 2(H_T^1-H_T^2) + K (A_T^1-A_T^2) \\
& \leq 2(M_T^1-M_T^2) + K (A_T^1-A_T^2).
\end{split}
\]
Since $M^i$ is martingale, also note that $\tau_1\wedge T \leq \tau_2\wedge T$, 
we have
\begin{equation} \label{first part}
\begin{split}
\mathbb{E}|M_T^1-M_T^2|^2 
& \leq  K \mathbb{E}[A_T^1-A_T^2] \\
&= K  \mathbb{E}\left[\int_t^{\tau_1\wedge T} h_s^1 ds - \int_t^{\tau_2\wedge T} h_s^2 ds \right] \\
& \leq K \mathbb{E}\left[\int_t^{\tau_1\wedge T} (h_s^1-h_s^2) ds \right] \\
& \leq K |s_1-s_2|. 
\end{split}
\end{equation}
Combining (\ref{first part}) and (\ref{second part}), we conclude that $\mathbb{E}\left[\sup_{[t,T]}|H_u^1-H_u^2|\right]  \leq K|s_1-s_2|$,  which  gives
\[
|v_0(t,x,s_1,p)-v_0(t,x,s_2,p)| \leq K x^\gamma|s_1-s_2|^{\frac{1}{2}}.
\]
Therefore, $v_0$ is continuous in $s$, uniformly in $t,p$.

\textbf{Step 3.} 
Fix $0<p_1<p_2<\infty$ and $t,x,s\in [0,T]\times [0,\infty)\times (0,\infty)$, by same technique as in Step 2, we can show
\[
|v_0(t,x,s,p_1)-v_0(t,x,s,p_2)| \leq K x^\gamma|p_1-p_2|^{\frac{1}{2}}.
\]
Therefore, $v_0$ is continuous in $p$, uniformly in $t,s$.

\textbf{Step 4.} 
For any $0\leq t_1 < t_2\leq T$ and $x,s,p\in [0,\infty)\times (0,\infty)^2$, by the definition of $v_0$ and the dynamic programming principle, $\forall\delta >0$, $\exists\pi(\delta)\in \mathcal{A}$ such that
\[ 
\begin{split}
v_0(t_1,x,s,p) - \delta & \leq \mathbb{E}\left[v_0\left(t_2,X_{t_2}^{t_1,x,s,p,\pi(\delta)}, S_{t_2}^{t_1,s}, P_{t_2}^{t_1,p}\right)\mathbb{I}_{\{ t_2 < \tau \}} + v_1\left(t_2,X_{t_2}^{t_1,x,\pi(\delta)}\right)\mathbb{I}_{\{ t_2 \geq \tau \}}\right] \\ 
& \leq v_0(t_1,x,s,p).
\end{split}
\]

Rearranging the order, we have
\[
\begin{split}
& \ \ \ \ \left|v_0(t_1,x,s,p) - v_0(t_2,x,s,p)\right| - \delta \\
& \leq \left|\mathbb{E}\left[v_0\left(t_2,X_{t_2}^{t_1,x,s,\pi(\delta)},S_{t_2}^{t_1,s},P_{t_2}^{t_1,p}\right)\mathbb{I}_{\{ t_2 < \tau \}} + v_1\left(t_2,X_{t_2}^{t_1,x,\pi(\delta)}\right)\mathbb{I}_{\{ t_2 \geq \tau \}}\right] - v_0(t_2,x,s,p)\right| \\
& \leq \mathbb{E}\left[\left|v_0\left(t_2,X_{t_2}^{t_1,x,s,\pi(\delta)},S_{t_2}^{t_1,s},P_{t_2}^{t_1,p}\right)\mathbb{I}_{\{ t_2 < \tau \}} - v_0(t_2,x,s,p)\right|\right] + \mathbb{E}\left[\left|v_1\left(t_2,X_{t_2}^{t_1,x,\pi(\delta)}\right)\mathbb{I}_{\{ t_2 \geq \tau \}}\right|\right]. 
\end{split}
\]

Using the Cauchy-Schwartz inequality, we have
\[
\begin{split}
\mathbb{E}\left[\left|v_1\left(t_2,X_{t_2}^{t_1,x,s,\pi(\delta)}\right)\mathbb{I}_{\{ t_2 \geq \tau \}}\right|\right] & \leq \mathbb{E}\left[\left|v_1\left(t_2,X_{t_2}^{t_1,x,s,\pi(\delta)}\right)\right|^2\right]^{1/2}\sqrt{\mathbb{P}\big( t_2 \geq \tau \big)} \\
& \leq K(1+x^\gamma)\sqrt{\mathbb{P}\big( t_2 \geq \tau \big)}, 
\end{split}
\]
which tends to 0 since
$\mathbb{P}\big( t_2 \geq \tau \big)\to0 $ as $t_2-t_1\to 0$.
  
Next we prove the first term $\mathbb{E}\left[\left|v_0\left(t_2,X_{t_2}^{t_1,x,s,p,\pi(\delta)},S_{t_2}^{t_1,s},P_{t_2}^{t_1,p}\right)\mathbb{I}_{\{ t_2 < \tau \}} - v_0(t_2,x,s,p)\right|\right]$ goes to zero as $t_2-t_1\rightarrow 0$.
\[
\begin{split}
& \ \ \ \ \mathbb{E}\left[\left|v_0\left(t_2,X_{t_2}^{t_1,x,s,p,\pi(\delta)},S_{t_2}^{t_1,s},P_{t_2}^{t_1,p}\right)\mathbb{I}_{\{ t_2 < \tau \}} - v_0(t_2,x,s,p)\right|\right] \\
& \leq \mathbb{E}\left[\left|\left(v_0\left(t_2,X_{t_2}^{t_1,x,s,p,\pi(\delta)},S_{t_2}^{t_1,s},P_{t_2}^{t_1,p}\right) - v_0\left(t_2,x,S_{t_2}^{t_1,s},P_{t_2}^{t_1,p}\right)\right)\mathbb{I}_{\{ t_2 < \tau \}}\right|\right] \\
& \ \ \ + \mathbb{E}\left[\left|\left(v_0\left(t_2,x,S_{t_2}^{t_1,s},P_{t_2}^{t_1,p}\right)- v_0(t_2,x,s,P_{t_2}^{t_1,p})\right)\mathbb{I}_{\{ t_2 < \tau \}}\right|\right] \\
& \ \ \ + \mathbb{E}\left[\left|\left(v_0\left(t_2,x,s,P_{t_2}^{t_1,p}\right)- v_0(t_2,x,s,p)\right)\mathbb{I}_{\{ t_2 < \tau \}}\right|\right] + \mathbb{E}\left[\left|v_0(t_2,x,s,p)\mathbb{I}_{\{ t_2\geq \tau \}} \right|\right].
\end{split}
\]

As shown in Step 1, $\left|v_0\left(t_2,X_{t_2}^{t_1,x,s,p,\pi(\delta)},S_{t_2}^{t_1,s},P_{t_2}^{t_1,p}\right) - v_0\left(t_2,x,S_{t_2}^{t_1,s},P_{t_2}^{t_1,p}\right)\right|\leq K \left| X_{t_2}^{t_1,x,s,p,\pi(\delta)}-x \right|^\gamma$, and by 
 (\ref{strong solution property}), 
\[
\mathbb{E}\left[ \left| X_{t_2}^{t_1,x,s,p,\pi(\delta)}-x \right|^2 \right] \leq 2x^2 + 2\mathbb{E}\left[ \left| X_{t_2}^{t_1,x,s,p,\pi(\delta)}\right|^2\right] < \infty.
\]
Therefore, $\left| X_{t_2}^{t_1,x,s,p,\pi(\delta)}-x \right|^\gamma$ is  uniformly integrable, and we can exchange the order of expectation and limit to get
\begin{eqnarray*}
&& \lim_{t_2-t_1\rightarrow 0} \mathbb{E}\left[\left|v_0\left(t_2,X_{t_2}^{t_1,x,s,p,\pi(\delta)},S_{t_2}^{t_1,s},P_{t_2}^{t_1,p}\right) - v_0\left(t_2,x,S_{t_2}^{t_1,s},P_{t_2}^{t_1,p}\right)\right| \mathbb{I}_{\{ t_2 < \tau \}}\right]  \\
&\leq& \mathbb{E}\left[ K \lim_{t_2-t_1\rightarrow 0}\left| X_{t_2}^{t_1,x,s,p,\pi(\delta)}-x \right|^\gamma \right] = 0.
\end{eqnarray*}

The same argument can be applied to the term $\mathbb{E}\left[\left|v_0\left(t_2,x,S_{t_2}^{t_1,s},P_{t_2}^{t_1,p}\right) - v_0(t_2,x,s,P_{t_2}^{t_1,p})| \mathbb{I}_{\{ t_2 < \tau \}}\right|\right]$ based on Step 2 and $\mathbb{E}\left[\left|v_0\left(t_2,x,s,P_{t_2}^{t_1,p}\right) - v_0(t_2,x,s,p)| \mathbb{I}_{\{ t_2 < \tau \}}\right|\right]$ based on Step 3, and we conclude that 
\[
\begin{split}
& \ \ \ \lim_{t_2-t_1\rightarrow 0} \mathbb{E}\left[\left|v_0\left(t_2,x,S_{t_2}^{t_1,s},P_{t_2}^{t_1,p}\right) - v_0(t_2,x,s,P_{t_2}^{t_1,p})| \mathbb{I}_{\{ t_2 < \tau \}}\right|\right]  \leq \mathbb{E}\left[ Kx^\gamma \lim_{t_2-t_1\rightarrow 0}\left| S_{t_2}^{t_1,s}-s \right|^{\frac{1}{2}} \right] = 0
\end{split}
\]
and
\[
\begin{split}
& \ \ \ \lim_{t_2-t_1\rightarrow 0} \mathbb{E}\left[\left|v_0\left(t_2,x,s,P_{t_2}^{t_1,p}\right) - v_0(t_2,x,s,p)| \mathbb{I}_{\{ t_2 < \tau \}}\right|\right]  \leq \mathbb{E}\left[ Kx^\gamma \lim_{t_2-t_1\rightarrow 0}\left| P_{t_2}^{t_1,p}-p \right|^{\frac{1}{2}} \right] = 0.
\end{split}
\]

The last term $|v_0(t_2,x,s,p)|  \mathbb{P}( t_2 \geq \tau) \leq K(1+x^\gamma)  \mathbb{P}( t_2 \geq \tau)$, which tends to zero when $t_2-t_1\rightarrow 0$. Therefore 
\[ \mathbb{E}\left[\left|v_0\left(t_2,X_{t_2}^{t_1,x,s,p,\pi(\delta)},S_{t_2}^{t_1,s},P_{t_2}^{t_1,p}\right)\mathbb{I}_{\{ t_2 < \tau \}} - v_0(t_2,x,s,p)\right|\right]\rightarrow 0 \]
as $t_2-t_1\rightarrow 0$ and we finally have
\[
\lim_{t_2-t_1\rightarrow 0} |v_0(t_1,x,s,p)-v_0(t_2,x,s,p)| \leq \delta.
\]
Since $\delta$ is arbitrary, we conclude that $v_0(t,x,s,p)$ is continuous in $t$.
Combining Steps 1,2,3,4, we  conclude that $v_0(t,x,s,p)$ is continuous in $[0,T]\times [0,\infty)\times (0,\infty)^2$.
\end{proof}

\subsection{Proof of Theorem~\ref{pre-default verification theorem}}\label{2.8}
\begin{proof}
For $\forall \pi\in\mathcal{A}$, define a new process $w\left(u,X_u^{t,x,s,z,\pi},S_u^{t,s},\mathbb{H}_u\right):=\sum_{\bar{z}\in I} w_{\bar{z}}\left(u,X_u^{t,x,s,z,\pi},S_u^{t,s}\right)\mathbb{I}_{\{\mathbb{H}_u=\bar{z}\}}$ where $X_u^{t,x,s,z,\pi}$ denotes the wealth process starting with $X_t=x, S_t=s, \mathbb{H}_t=z$ associated with control process $\pi$, and $S_u^{t,s}$ denotes the prices of surviving stocks at time $u$ starting with $S_t=s$. 

As $w_{\bar{z}}$ is smooth in $(t,x,s)$ for $\forall \bar{z}\in I$, we can apply Ito\rq{}s formula to $w$ and get for any time $u\in [t,T]$ 
\begin{equation*}
w(u,X_u^{t,x,s,z\pi},S_u^{t,s},\mathbb{H}_u) = w_z(t,x,s) + \int_t^u \sum_{\bar{z}\in I} \mathcal{L}^\pi w_{\bar{z}}(\bar u,X_{\bar u}^\pi,S_{\bar u}) \mathbb{I}_{\{\mathbb{H}_{\bar u}=\bar{z}\}}  \ d\bar u + M_u - M_t,
\end{equation*}
where $\mathcal{L}^\pi w_{\bar{z}}$ is defined in (\ref{e2.4}), and $M$ is a local martingale defined by
\begin{eqnarray}
M_u &:= &\sum_{\bar{z}\in I}\left( \sum_{i\in I_{\bar z}} \int_t^u \sigma_i S_{\bar{u}}^i \frac{\partial w_{\bar{z}}}{\partial s_i} \mathbb{I}_{\{\mathbb{H}_{\bar u}=\bar{z}\}} dW_{\bar{u}}^i + \int_t^u \pi_{\bar u}^T \sigma X_{\bar u}^\pi \frac{\partial w_{\bar{z}}}{\partial x}\mathbb{I}_{\{\mathbb{H}_{\bar u}=\bar{z}\}} dW_{\bar u} \right) \nonumber \\
&& {} + \sum_{\bar{z}\in I}\left( \sum_{i\in I_{\bar z}} \int_t^u \left( w_{\bar{z}^i}\left(\bar{u},X_{\bar{u}-}\left(1-\sum_{j=1}^N L_{ji}\pi_{\bar{u}}^j\right),S_{\bar{u}-}^i\right) - w_{\bar{z}}(\bar{u},X_{\bar{u}-},S_{\bar{u}-})\right) \left(dH_{\bar{u}}^i - h^i_{\bar z}(S_{\bar{u}-})\mathbb{I}_{\{\mathbb{H}_{\bar u}=\bar{z}\}}d\bar{u}\right)\right) \nonumber
\end{eqnarray}

Since $w_{\bar{z}}$ satisfies the HJB equation (\ref{General pre-default HJB equation}), we have $\mathcal{L}^\pi w_{\bar{z}}\leq 0$. 
Define stopping times 
$$\widetilde{\tau}_n:=\inf\left\{u\geq t: \left|X_u^{t,x,s,z,\pi}-x\right| + \sum_{i\in I_z}\left|S_u^i-s_i\right| \geq n \right\}\wedge (T-1/n),$$
 then $M_{u\wedge\widetilde{\tau}_n}$ is a martingale due to the boundedness of control set $A$ and values and derivatives of $w_{\bar{z}}$. Letting $u=T$ and taking expectation on both sides, we have
\[ \mathbb{E}\left[w\left(\widetilde{\tau}_n,X_{\widetilde{\tau}_n}^\pi,S_{\widetilde{\tau}_n},\mathbb{H}_{\widetilde{\tau}_n}\right)\right] 
\leq w_z(t,x,s) \] 
with equality if $\pi=\widehat{\pi}$.
 Next we show that 
\begin{equation} \label{e4.7}
\lim_{n\rightarrow \infty} \mathbb{E}\left[w_{\bar{z}}\left( \widetilde{\tau}_n,X_{\widetilde{\tau}_n}^\pi,S_{\widetilde{\tau}_n}\right)\mathbb{I}_{\{\mathbb{H}_{\widetilde{\tau}_n}=\bar{z}\}}\right] =\mathbb{E}\left[w_{\bar{z}}\left(T,X_T^\pi,S_T\right)\mathbb{I}_{\{\mathbb{H}_T=\bar{z}\}}\right]
=\mathbb{E}\left[U\left(X_T^\pi\right)\mathbb{I}_{\{\mathbb{H}_T=\bar{z}\}}\right],
\end{equation}
for $\forall \bar{z}\in I$.
Since $|w_{\bar{z}}(t,x,s)| \leq K(1+x^\gamma)$, also noting (\ref{strong solution property}), we have
\[
\mathbb{E}\left[\left|w_{\bar{z}}\left( \widetilde{\tau}_n,X_{\widetilde{\tau}_n}^\pi,S_{\widetilde{\tau}_n}\right)\mathbb{I}_{\{\mathbb{H}_{\widetilde{\tau}_n}=\bar{z}\}}\right|^{\alpha}\right] \leq K\left(1 + \mathbb{E}\left[\left(X_{\widetilde{\tau}_n}^\pi\right)^{\alpha \gamma}\right] \right)\leq K(1+x^{\alpha \gamma})<\infty
\]
for  any $\alpha>1$.  Since $w_{\bar{z}}\left( \widetilde{\tau}_n,X_{\widetilde{\tau}_n}^\pi,S_{\widetilde{\tau}_n}\right)\mathbb{I}_{\{\mathbb{H}_{\widetilde{\tau}_n}=\bar{z}\}}$ is  uniformly integrable, we can take the limit under the  expectation to get (\ref{e4.7}). This shows that 
 $ \mathbb{E}[U(X_T^\pi)] \leq w_z(t,x,s)$ with equality if $\pi=\widehat{\pi}$. 
Furthermore, SDE (\ref{wealth process}) 
admits a unique strong solution by the assumption, therefore, $w_z$ coincides with the  value function $v_z$ and $\widehat{\pi}$ is the optimal control process.
\end{proof}

\subsection{Proof of Theorem~\ref{vis}}\label{2.11a} 
\begin{lemma} \label{Lemma for continuity of F in viscosity solution proof}
Denote by $R_z:=R_z(t,x,s)$  the following function
\begin{eqnarray}
R_z &:= &\sup_{\pi\in A}\left \{ \theta^T\pi x \frac{\partial w_z}{\partial x} + \frac{1}{2}\pi^T\Sigma\pi x^2 \frac{\partial^2 w_z}{\partial x^2} + \sum_{i\in I_z}\rho_i^T\sigma\pi\sigma_i xs_i \frac{\partial^2 w_z}{\partial x\partial s_i} + \sum_{i\in I_z} h^i_z(s)  w_{z^i}\left(t,x\left(1-\sum_{j=1}^N L_{ji}\pi^j\right),s^i\right)\right \}, \nonumber
\end{eqnarray}
where $w_z\in C^{1,2,\ldots,2}$ for $\forall z\in I$. Then $R_z$ is continuous in $(t,x,s)$.
\end{lemma}
\begin{proof}
Let $z\in I$ and the point $(\bar{t},\bar{x},\bar{s})\in [0,T)\times (0,\infty)^{N_z+1}$ and $B_{\eta}(\bar{t},\bar{x},\bar{s})$ be the ball with center $(\bar{t},\bar{x},\bar{s})$ and radius $\eta$. By the definition of supremum function, for any $\delta>0$, there exists a control $\pi\in A$ such that
\begin{eqnarray}
R_z(\bar{t},\bar{x},\bar{s}) - \delta &\leq & \theta^T\pi \bar{x} \frac{\partial w_z}{\partial x}(\bar{t},\bar{x},\bar{s}) + \frac{1}{2}\pi^T\Sigma\pi \bar{x}^2 \frac{\partial^2 w_z}{\partial x^2}(\bar{t},\bar{x},\bar{s}) + \sum_{i\in I_z}\rho_i^T\sigma\pi\sigma_i \bar{x}\bar{s_i} \frac{\partial^2 w_z}{\partial x\partial s_i}(\bar{t},\bar{x},\bar{s}) \nonumber \\
&&{}  + \sum_{i\in I_z} h^i_z(\bar{s})  w_{z^i}\left(\bar{t},\bar{x}\left(1-\sum_{j=1}^N L_{ji}\pi^j\right),\bar{s}^i\right), \label{1}
\end{eqnarray}

For any point $(t,x,s)\in B_{\eta}(\bar{t},\bar{x},\bar{s})$, we have
\begin{eqnarray}
R_z(t,x,s) &\geq & \theta^T\pi x \frac{\partial w_z}{\partial x}(t,x,s) + \frac{1}{2}\pi^T\Sigma\pi x^2 \frac{\partial^2 w_z}{\partial x^2}(t,x,s) + \sum_{i\in I_z}\rho_i^T\sigma\pi\sigma_i x s_i \frac{\partial^2 w_z}{\partial x\partial s_i}(t,x,s) \nonumber \\
&&{}  + \sum_{i\in I_z} h^i_z(s)  w_{z^i}\left(t,x\left(1-\sum_{j=1}^N L_{ji}\pi^j\right),s^i\right), \label{2}
\end{eqnarray}

Subtracting (\ref{2}) from (\ref{1}), we have
\begin{eqnarray}
R_z(\bar{t},\bar{x},\bar{s}) - R_z(t,x,s) - \delta &\leq & \theta^T\pi \left( \bar{x} \frac{\partial w_z}{\partial x}(\bar{t},\bar{x},\bar{s}) - x \frac{\partial w_z}{\partial x}(t,x,s)\right) + \frac{1}{2}\pi^T\Sigma\pi \left(\bar{x}^2 \frac{\partial^2 w_z}{\partial x^2}(\bar{t},\bar{x},\bar{s}) - x^2 \frac{\partial^2 w_z}{\partial x^2}(t,x,s)\right) \nonumber \\
&& {} + \sum_{i\in I_z}\rho_i^T\sigma\pi\sigma_i \left(\bar{x}\bar{s_i} \frac{\partial^2 w_z}{\partial x\partial s_i}(\bar{t},\bar{x},\bar{s}) - x s_i \frac{\partial^2 w_z}{\partial x\partial s_i}(t,x,s) \right) \nonumber \\
&&{}  + \sum_{i\in I_z} \left(h^i_z(\bar{s})  w_{z^i}\left(\bar{t},\bar{x}\left(1-\sum_{j=1}^N L_{ji}\pi^j\right),\bar{s}^i\right) - h^i_z(s)  w_{z^i}\left(t,x\left(1-\sum_{j=1}^N L_{ji}\pi^j\right),s^i\right) \right). \nonumber
\end{eqnarray}

Taking the limit superior and then letting $\delta$ tend to 0, we have
\begin{equation} \label{4}
R_z(\bar{t},\bar{x},\bar{s}) \leq \liminf_{(t,x,s)\rightarrow (\bar{t},\bar{x},\bar{s})} R_z(t,x,s).
\end{equation}
Similarly, using the smoothness of $w_z$ and boundedness of $\pi$ and $h_z^i$ for $\forall z\in I$ and $i\in \{1,\ldots, N\}$, we can get
\begin{equation} \label{5}
R_z(\bar{t},\bar{x},\bar{s}) \geq \limsup_{(t,x,s)\rightarrow (\bar{t},\bar{x},\bar{s})} R_z(t,x,s).
\end{equation}
Inequalities (\ref{4}) and (\ref{5}) imply that $R_z(t,x,s)$ is continuous in $(t,x,s)$ for $\forall z\in I$.
\end{proof}

Combining Lemma~\ref{Lemma for continuity of F in viscosity solution proof} and the smoothness of $w$, we conclude that $F_z$ given in Definition \ref{viscosity} is continuous in $(t,x,s)$. Based on this result, we prove the value function $v:=(v_z)_{z\in I}$ is a viscosity solution to the PDE system (\ref{General pre-default HJB equation}).

\begin{proposition} \label{viscosity solution}
The value function $v:=(v_z)_{z\in I}$ is a  viscosity supersolution to equation (\ref{General pre-default HJB equation}) on $[0,T)\times (0,\infty)^{N+1}$.
\end{proposition}
\begin{proof}
Let $\bar{z}\in I$, $(\bar{t},\bar{x},\bar{s})\in [0,T)\times (0,\infty)^{N_{\bar z}+1}$ and $\varphi:=(\varphi_z)_{z\in I}\in C^{1,2,...,2}\left([0,T)\times (0,\infty)^{N_z+1}\right)$ be tuple of test functions such that
\begin{equation} \label{super}
0=\left((v_{\bar{z}})_{*}-\varphi_{\bar{z}}\right)(\bar{t},\bar{x},\bar{s}) = \min_{[0,T)\times (0,\infty)^{N_{\bar z}+1}} \left((v_{\bar z})_{*}-\varphi_{\bar z}\right)(t,x,s),
\end{equation}
and $(v_z)_* \geq \varphi_z$ for $\forall z\in I$ on $[0,T)\times (0,\infty)^{N_z+1}$.

By definition of $(v_{\bar{z}})_{*}$, there exists a sequence $(t_m,x_m,s_m)$ in $[0,T)\times (0,\infty)^{N_{\bar z}+1}$ such that 
\[ 
(t_m,x_m,s_m)\rightarrow (\bar{t},\bar{x},\bar{s}) \ \ \text{and} \ \ v_{\bar{z}}(t_m,x_m,s_m)\rightarrow (v_{\bar{z}})_{*}(\bar{t},\bar{x},\bar{s}),
\]
when $m$ goes to infinity. By the continuity of $\varphi_{\bar z}$ and by (\ref{super}) we also have that
\[
\gamma_m:=v_{\bar{z}}(t_m,x_m,s_m)-\varphi_{\bar{z}}(t_m,x_m,s_m) \rightarrow 0,
\]
when $m$ goes to infinity.

Let $\pi\in \mathcal{A}$ be a constant control process and $B_{\eta}(x_m,s_m)\in (0,\infty)^{N_{\bar z}+1}$ be the ball with center $(x_m,s_m)$ and radius $\eta>0$. Note that when $m$ is large enough, $(x_m,s_m)\in B_{\eta}(\bar{x},\bar{s})$, thus $\forall (x,s) \in B_{\eta}(x_m,s_m)$, we have $(x,s)\in B_{2\eta}(\bar{x},\bar{s})$. We denote by $X_u^{t_m,x_m}$ the associated controlled wealth process. Let $\tau_m^{\pi}$ be the stopping time given by
\[
\tau_m^{\pi}:=\inf \left \{ u\in [t_m,T):\left(X_u^{t_m,x_m},S_u^{t_m,s_m}\right) \notin B_{\eta}(x_m,s_m) \right \}.
\]
Let $(h_m)$ be a strictly positive sequence such that
\[
h_m \rightarrow 0 \ \ \text{and} \ \ \frac{\gamma_m}{h_m} \rightarrow 0,
\]
when $m$ goes to infinity. Then we can define a stopping time $\theta_m$ give by $\theta_m:=\tau_m^{\pi}\wedge (t_m+h_m) \wedge \widetilde{\tau}_m$ where $\widetilde{\tau}_m$ is the first default time of the surviving stocks, starting from $t_m$.

Next we use the weak dynamic programming principle (weak-DPP) proved in Bouchard and Touzi (2011), that is,
\[
v_{\bar{z}}(t,x,s)\geq \mathbb{E}\left[\sum_{z\in I}(v_z)_{*}\left( \theta, X_{\theta}^{t,x,s,\bar{z}},S_{\theta}^{t,s}\right)\mathbb{I}_{\{\mathbb{H}_{\theta}=z\}}\right],
\]
for any $\mathcal{G}$-measurable stopping time $\theta \in [t,T]$ such that $X_{\theta}$ and $S_{\theta}$ are $\mathbb{L}^\infty$-bounded.

Since under stopping time $\theta_m$, the processes $S$ and $X$ are both bounded, we can apply above weak dynamic programming principle (weak-DPP) for $v_{\bar{z}}(t_m,x_m,s_m)$ to $\theta_m$ and get
\[
v_{\bar{z}}(t_m,x_m,s_m)\geq \mathbb{E}\left[\sum_{z\in I}(v_z)_{*}\left( \theta_m, X_{\theta_m}^{t_m,x_m,s_m,\bar{z}},S_{\theta_m}^{t_m,s_m}\right)\mathbb{I}_{\{\mathbb{H}_{\theta_m}=z\}}\right].
\]

Equation (\ref{super}) implies $(v_z)_{*}\geq \varphi_z$ for $\forall z\in I$, thus
\[
\varphi_{\bar{z}}(t_m,x_m,s_m) + \gamma_m\geq \mathbb{E}\left[\sum_{z\in I}\varphi_z\left( \theta_m, X_{\theta_m}^{t_m,x_m,s_m,\bar{z}},S_{\theta_m}^{t_m,s_m}\right)\mathbb{I}_{\{\mathbb{H}_{\theta_m}=z\}}\right].
\]
Applying Ito's formula to the whole term in bracket, we obtain
\begin{equation} \label{super1}
\frac{\gamma_m}{h_m} - \mathbb{E}\left[ \frac{1}{h_m}\int_{t_m}^{\theta_m} \sum_{z\in I} \mathcal{L}^\pi \varphi_{z}(u,X_{u}^\pi,S_{u}) \mathbb{I}_{\{\mathbb{H}_u=z\}} du \right] \geq 0
\end{equation}
after noting that the stochastic integral term cancels out by taking expectations since the integrand is bounded. Note that $\mathcal{L}^\pi\varphi_z\left(u,X_u^{t_m,x_m},S_u^{t_m,s_m}\right)$ is defined the same as (\ref{e2.4}).

Next we investigate the stopping time $\theta_m$ when $m$ is large enough. Firstly, for the stopping time $\tau_m^{\pi}$, denoting $E_m:=\left\{\widetilde{\tau}_m>t_m+h_m\right\}$, we have
\[
\begin{split}
& \ \ \ \ \mathbb{P}\left(\tau_m^{\pi}\leq t_m+h_m \mid E_m \right) \\
& = \mathbb{P}\left(\sup_{t\in [t_m,t_m+h_m]} \left(\left|X_t^{t_m,x_m}-x_m\right|^2 + \sum_{i\in I_{\bar{z}}}\left|S_t^{i,t_m,s_m}-s_{m,i}\right| ^2 \right)\geq \eta^2 \mid E_m \right) \\
&\leq \mathbb{P}\left(\sup_{t\in [t_m,t_m+h_m]} \left|X_t^{t_m,x_m}-x_m\right|^2 \geq \frac{\eta^2}{N_{\bar z}+1} \mid E_m \right) + \sum_{i\in I_{\bar z}} \mathbb{P}\left(\sup_{t\in [t_m,t_m+h_m]} \left| S_t^{i,t_m,s_m}-s_{m,i}\right|^2 \geq \frac{\eta^2}{N_{\bar z}+1} \mid E_m \right)  \\
&\leq \frac{N_{\bar z}+1}{\eta^2}\left ( \mathbb{E}\left[\sup_{t\in [t_m,t_m+h_m]} \left|X_t^{t_m,x_m}-x_m\right|^2 \mid E_m \right] + \sum_{i\in I_{\bar z}}\mathbb{E}\left[\sup_{t\in [t_m,t_m+h_m]} \left| S_t^{i,t_m,s_m}-s_{m,i}\right|^2 \mid E_m \right] \right).   
\end{split}
\]
By Pham (2009), Page 67,  each term in the bracket converges to zero as $m\rightarrow \infty$, which gives
\[
\lim_{m\rightarrow \infty} \mathbb{P}\left(\tau_m^{\pi}\leq t_m+h_m \mid E_m\right) = 0.
\]
By definition of conversion time $\widetilde{\tau}_m$, we have
\[
\mathbb{P}\left(E_m^C\right) = 1 - \mathbb{E}\left[e^{-\int_{t_m}^{t_m+h_m} \sum_{i\in I_{\bar z}} h_{\bar z}^i(S_u) du} \right] \leq 1-e^{-K h_m}
\]
due to the boundedness of intensity function $h_{\bar z}^i$. Thus
\[
\lim_{m\rightarrow \infty} \mathbb{P}(E_m^C) = 0.
\]
Finally for the stopping time $\tau_m^{\pi}\wedge \widetilde{\tau}_m$, we have
\begin{eqnarray}
\mathbb{P}\left(\tau_m^{\pi}\wedge \widetilde{\tau}_m\leq t_m+h_m\right) &\leq & \mathbb{P}\left(\tau_m^{\pi}\leq t_m+h_m\right) + \mathbb{P}\left(E_m^C\right) \nonumber \\
&= & \mathbb{P}\left(\tau_m^{\pi}\leq t_m+h_m, E_m\right) + \mathbb{P}\left(\tau_m^{\pi}\leq t_m+h_m, E_m^C\right) + \mathbb{P}\left(E_m\right) \nonumber \\
&\leq & \mathbb{P}\left(\tau_m^{\pi}\leq t_m+h_m \mid E_m\right) + 2\mathbb{P}\left(E_m^C\right).
\end{eqnarray}
Combining above results, we get
\[
\lim_{m\rightarrow \infty} \mathbb{P}\left(\tau_m^{\pi}\wedge \widetilde{\tau}_m\leq t_m+h_m\right) = 0.
\]

We now estimate 
\[
\begin{split}
-\frac{\gamma_m}{h_m} & \leq 
\mathbb{E}\left[ \frac{1}{h_m}\int_{t_m}^{\theta_m} - \sum_{z\in I} \mathcal{L}^\pi \varphi_{z}(u,X_{u}^\pi,S_{u}) \mathbb{I}_{\{\mathbb{H}_u=z\}} du \right] \\
& \leq 
\mathbb{E}\left[ \frac{1}{h_m}\int_{t_m}^{t_m+h_m} - \sum_{z\in I} \mathcal{L}^\pi \varphi_{z}(u,X_{u}^\pi,S_{u}) \mathbb{I}_{\{\mathbb{H}_u=z\}} du \mid \tau_m^{\pi}\wedge \widetilde{\tau}_m > t_m+h_m \right] \mathbb{P}\left(\tau_m^{\pi}\wedge \widetilde{\tau}_m> t_m+h_m\right) \\
& \ \ \ + \mathbb{E}\left[ \frac{1}{h_m}\int_{t_m}^{\tau_m^{\pi}\wedge \widetilde{\tau}_m} - \sum_{z\in I} \mathcal{L}^\pi \varphi_{z}(u,X_{u}^\pi,S_{u}) \mathbb{I}_{\{\mathbb{H}_u=z\}} du \mid \tau_m^{\pi}\wedge \widetilde{\tau}_m\leq t_m+h_m \right] \mathbb{P}\left(\tau_m^{\pi}\wedge \widetilde{\tau}_m\leq t_m+h_m\right) \\
& \leq \mathbb{E}\left[ \frac{1}{h_m}\int_{t_m}^{t_m+h_m} - \mathcal{L}^\pi\varphi_{\bar{z}}\left(u,X_u^{t_m,x_m},S_u^{t_m,s_m}\right)du \mid \tau_m^{\pi}\wedge \widetilde{\tau}_m > t_m+h_m \right] + K\mathbb{P}\left(\tau_m^{\pi}\wedge \widetilde{\tau}_m\leq t_m+h_m\right).
\end{split}
\]
By the mean value theorem and dominated convergence theorem, taking limit on both sides of the inequality, we have
\[
-\mathcal{L}^\pi\varphi_{\bar{z}}(\bar{t},\bar{x},\bar{s})\geq 0,
\]
which implies 
\[
F_{\bar{z}}\left(\bar{t},\bar{x},\bar{s},\varphi,\nabla_{(t,x,s)}\varphi_{\bar{z}},\nabla^2_{(x,s)}\varphi_{\bar{z}}\right ) \geq 0,
\]
due to the arbitrariness of $\pi\in A$. 
\end{proof}

\begin{proposition} \label{viscosity solution}
The value function $v:=(v_z)_{z\in I}$ is a  viscosity subsolution to equation (\ref{General pre-default HJB equation}) on $[0,T)\times (0,\infty)^{N+1}$.
\end{proposition}
\begin{proof}
Let $\bar{z}\in I$, $(\bar{t},\bar{x},\bar{s})\in [0,T)\times (0,\infty)^{N_{\bar z}+1}$ and $\varphi:=(\varphi_z)_{z\in I}\in C^{1,2,...,2}\left([0,T)\times (0,\infty)^{N_z+1}\right)$ be tuple of test functions such that
\begin{equation} \label{sub}
0=\left((v_{\bar{z}})^{*}-\varphi_{\bar{z}}\right)(\bar{t},\bar{x},\bar{s}) = \max_{[0,T)\times (0,\infty)^{N_{\bar z}+1}} \left((v_{\bar z})^{*}-\varphi_{\bar z}\right)(t,x,s),
\end{equation}
and $(v_z)^* \leq \varphi_z$ for $\forall z\in I$ on $[0,T)\times (0,\infty)^{N_z+1}$.

We prove the result by contradiction. Assume on the contrary that
\[
F_{\bar z}\left(\bar{t},\bar{x},\bar{s},\varphi,\nabla_{(t,x,s)}\varphi_{\bar z},\nabla^2_{(x,s)}\varphi_{\bar z}\right ) > 0
\]
Then by the continuity of $F_{\bar z}$, there exists $\delta>0$ and $\eta>0$ such that
\[
F_{\bar z}\left(t,x,s,\varphi,\nabla_{(t,x,s)}\varphi,\nabla^2_{(x,s)}\varphi\right ) = -\sup_{\pi\in A}\mathcal{L}^\pi\varphi_{\bar z}(t,x,s) > \delta
\]
for $(t,x,s)\in B_{\eta}(\bar{t},\bar{x},\bar{s})$.
By definition of $(v_{\bar z})^{*}$, there exists a sequence $(t_m,x_m,s_m)$ taking values in $B_{\frac{\eta}{2}}(\bar{t},\bar{x},\bar{s})$ such that 
\[ 
(t_m,x_m,s_m)\rightarrow (\bar{t},\bar{x},\bar{s}) \ \ \text{and} \ \ v_{\bar z}(t_m,x_m,s_m)\rightarrow (v_{\bar z})^{*}(\bar{t},\bar{x},\bar{s}),
\]
when $m$ goes to infinity. By the continuity of $\varphi_{\bar z}$ and by (\ref{sub}) we also have that
\[
\gamma_m:=v_{\bar z}(t_m,x_m,s_m)-\varphi_{\bar z}(t_m,x_m,s_m) \rightarrow 0,
\]
when $m$ goes to infinity.

We denote by $X_u^{t_m,x_m,\pi}$ the controlled wealth process associated with control process $\pi\in \mathcal{A}$. Let $\tau_m^{\pi^m}$ be the stopping time given by
\[
\tau_m^{\pi}:=\inf \left \{ u\in [t_m,T):\left(u,X_u^{t_m,x_m,\pi},S_u^{t_m,s_m}\right) \notin B_{\frac{\eta}{2}}(t_m,x_m,s_m) \right \}.
\]
Let $(h_m)$ be a strictly positive sequence such that
\[
h_m \rightarrow 0 \ \ \text{and} \ \ \frac{\gamma_m}{h_m} \rightarrow 0,
\]
when $m$ goes to infinity. Then we can define a stopping time $\theta_m$ give by $\theta_m:=\tau_m^{\pi}\wedge (t_m+h_m) \wedge \widetilde{\tau}_m$ where $\widetilde{\tau}_m$ is the first default time of surviving stocks starting from $t_m$.

Next we use the weak dynamic programming principle (weak-DPP) proved in Bouchard and Touzi (2011), that is, for any $\epsilon>0$, there exists a control process $\pi\in \mathcal{A}$ such that 
\[
v_{\bar z}(t,x,s)-\epsilon \leq \mathbb{E}\left[\sum_{z\in I}(v_z)^{*}\left( \theta, X_{\theta}^{t,x,s,\bar{z},\pi},S_{\theta}^{t,s}\right)\mathbb{I}_{\{\mathbb{H}_{\theta}=z\}} \right],
\]
for any $\mathcal{G}$--stopping time $\theta \in [t,T]$.

We apply above weak dynamic programming principle (weak-DPP) for $v_{\bar z}(t_m,x_m,s_m)$ to $\theta_m$ and get for $\epsilon=\delta\frac{h_m}{2}>0$, there exists $\pi\in \mathcal{A}$ such that
\begin{eqnarray}
v_{\bar z}(t_m,x_m,s_m)-\delta\frac{h_m}{2} &\leq &\mathbb{E}\left[\sum_{z\in I}(v_z)^{*}\left( \theta_m, X_{\theta_m}^{t_m,x_m,s_m,\bar{z},\pi},S_{\theta_m}^{t_m,s_m}\right)\mathbb{I}_{\{\mathbb{H}_{\theta_m}=z\}}\right]. \nonumber
\end{eqnarray}
Equation (\ref{sub}) implies $(v_z)^{*}\leq \varphi_z$ for $\forall z\in  I$, thus
\begin{eqnarray}
\varphi_{\bar z}(t_m,x_m,s_m) + \gamma_m -\delta\frac{h_m}{2} &\leq & \mathbb{E}\left[ \sum_{z\in I}\varphi_z\left( \theta_m, X_{\theta_m}^{t_m,x_m,s_m,\bar{z},\pi},S_{\theta_m}^{t_m,s_m}\right)\mathbb{I}_{\{\mathbb{H}_{\theta_m}=z\}} \right]. \nonumber
\end{eqnarray}

Applying Ito's formula to the whole term in bracket, we obtain
\begin{equation} \label{sub1}
\begin{split}
\frac{\gamma_m}{h_m} - \frac{\delta}{2} & \leq \mathbb{E}\left[ \frac{1}{h_m}\int_{t_m}^{\theta_m} \sum_{z\in I} \mathcal{L}^\pi \varphi_{z}(u,X_{u}^\pi,S_{u}) \mathbb{I}_{\{\mathbb{H}_u=z\}} du \right] \\
& \leq \mathbb{E}\left[ \frac{1}{h_m}\int_{t_m}^{t_m+h_m} \mathcal{L}^\pi \varphi_{\bar z}(u,X_{u}^\pi,S_{u}) du \mid \tau_m^\pi\wedge \widetilde{\tau}_m > t_m+h_m \right] + K\mathbb{P}\left(\tau_m^\pi\wedge \widetilde{\tau}_m \leq t_m+h_m\right)
\end{split}
\end{equation}
after noting that the stochastic integral term cancels out by taking expectations since the integrand is bounded. By the similar technique as the supersolution proof, we can show that $\mathbb{P}\left(\tau_m^\pi\wedge \widetilde{\tau}_m \leq t_m+h_m\right)\rightarrow 0$ as $m\rightarrow \infty$. 

Since $\left(u,X_u^{t_m,x_m,\pi^m},S_u^{t_m,s_m}\right)\in B_\eta(\bar{t},\bar{x},\bar{s})$ in $[t_m,t_m+h_m]$ if $\tau_m^\pi\wedge \widetilde{\tau}_m > t_m+h_m$, we have 
$$
\mathcal{L}^\pi\varphi_{\bar z}\left(u,X_u^{t_m,x_m,\pi^m},S_u^{t_m,s_m}\right)<-\delta
$$
in $[t_m,t_m+h_m]$. Thus
\[
\frac{\gamma_m}{h_m} - \frac{\delta}{2} \leq \mathbb{E}\left[ \frac{1}{h_m}\int_{t_m}^{t_m+h_m} -\delta du \right] + K\mathbb{P}\left(\tau_m^\pi\wedge \widetilde{\tau}_m \leq t_m+h_m\right).
\]
Then we obtain
\[
\lim_{m\rightarrow \infty}\frac{\gamma_m}{h_m} - \frac{\delta}{2} \leq -\delta,
\]
which implies $0\leq -\frac{\delta}{2}$. We thus get the desired contradiction with $\delta>0$.
\end{proof}

\subsection{Proof of Theorem~\ref{vis_unique}}\label{2.11} 
 To prove the comparison principle, we need an alternative definition of viscosity solution in terms of the notions of semijets defined as below.

\begin{definition}
For $z\in I$, given $w_z$ a function on $[0,T)\times (0,\infty)^{N_z+1}$, the superjet of $w_z$ at $(t,x,s)\in [0,T)\times (0,\infty)^{N_z+1}$ is defined by:
\[
\begin{split}
\mathcal{P}^{1,2,\ldots,2,+} w_z(t,x,s) =\Big\{ & (R,q,Q)\in \mathbb{R}\times \mathbb{R}^{N_z+1}\times \mathcal{S}^{(N_z+1)\times(N_z+1)} \ \text{such that} \\
& w_z(t^{\prime},x^{\prime},s^{\prime})\leq w_z(t,x,s)+R(t^{\prime}-t)+\langle q, X^{\prime}-X\rangle+\frac{1}{2}\langle Q(X^{\prime}-X), X^{\prime}-X\rangle \\
& \ \ \ \ \ \ \ \ \ \ \ \ \ \ \ \ \ \ +o(|t^{\prime}-t|^2+|X^{\prime}-X|^2) \Big\},
\end{split}
\]
where $X=(x,s)$, $X^{\prime}=(x^{\prime},s^{\prime})$, and the bracket $\langle\cdot,\cdot\rangle$ is the inner product of two vectors. We define its closure $\bar{\mathcal{P}}^{1,2,\ldots,2,+}w_z(t,x,s)$ as the set of elements $(R,q,Q)\in \mathbb{R}\times \mathbb{R}^{N_z+1}\times \mathcal{S}^{(N_z+1)\times(N_z+1)}$ for which there exists a sequence $(t_m,X_m,R_m,q_m,P_m)\in [0,T)\times (0,\infty)^{N_z+1}\times\mathcal{P}^{1,2,\ldots,2,+}w_z(t,X)$ satisfying $(t_m,X_m,R_m,q_m,Q_m)\rightarrow (t,X,R,q,Q)$. We also define the subjets
\[ \mathcal{P}^{1,2,\ldots,2,-}w_z(t,x,s)=-\mathcal{P}^{1,2,\ldots,2,+}(-w_z)(t,x,s),\ \ \bar{\mathcal{P}}^{1,2,\ldots,2,-}w_z(t,x,s)=-\bar{\mathcal{P}}^{1,2,\ldots,2,+}(-w_z)(t,x,s). \]
\end{definition}

By standard arguments, one has an equivalent definition of viscosity solutions in terms of semijets: $w:=(w_z)_{z\in I}$ is a viscosity subsolution (resp. supersolution) to (\ref{General pre-default HJB equation}) at $(t,x,s)\in [0,T)\times (0,\infty)^{N+1}$ if and only if for all $z\in I$ and $(R,q,Q)\in \bar{\mathcal{P}}^{1,2,\ldots,2,+}w_z(t,x,s)$ (resp. $\bar{\mathcal{P}}^{1,2,\ldots,2,-}w_z(t,x,s)$).
\[ F_z\left(t,x,s,w,(R,q),Q\right) \leq (resp. \geq) \ 0. \]

We can now state and prove  the following comparison principle which gives rise to the uniqueness of viscosity solution.

\begin{proposition} \label{Comparison principle}
Let $W:=(W_z)_{z\in I}$ (resp. $V:=(V_z)_{z\in I}$) be a u.s.c. viscosity subsolution (resp. l.s.c. viscosity supersolution) of (\ref{General pre-default HJB equation}) on $[0,T)\times (0,\infty)^{N+1}$ and satisfy the growth condition
$|W_z(t,x,s)|,  |V_z(t,x,s)| \leq K(1+x^\gamma)$,  
 the terminal relation $W_z(T,x,s) \leq V_z(T,x,s)$, and the boundary relations  $W_z(t,x,s)\leq V_z(t,x,s)$ on the boundary of $[0,\infty)^{N_z+1}$ for $\forall z\in I$.
Then we have $W_z\leq V_z$ for $\forall z\in I$ on $[0,T]\times [0,\infty)^{N_z+1}$.
\end{proposition}

\begin{proof} We prove the result in several steps.

\textbf{Step 1.} Let $\widetilde{W}_z=e^{\Gamma t}W_z$ and $\widetilde{V}_z=e^{\Gamma t}V_z$ for constant $\Gamma>0$, then a straightforward calculation shows that $\widetilde{W}$ (resp. $\widetilde{V}$) is a subsolution (resp. supersolution) of
$$
- \sup_{\pi\in A} \widetilde{\mathcal{L}}^\pi w_z(t,x,s) = 0, \ \ \text{on} \ \ [0,T)\times (0,\infty)^{N_z+1},
$$
for $z\in I$, where $\widetilde{\mathcal{L}}^\pi$ is given by
\begin{eqnarray}
\widetilde{\mathcal{L}}^\pi w_z(t,x,s) &= &\frac{\partial w_z}{\partial t} + (r+\theta^T\pi)x \frac{\partial w_z}{\partial x}  + \sum_{i\in I_z}\mu_is_i \frac{\partial w_z}{\partial s_i} + \frac{1}{2}\pi^T\Sigma\pi x^2 \frac{\partial^2 w_z}{\partial x^2} + \frac{1}{2}\sum_{i\in I_z}\sigma_i^2s_i^2 \frac{\partial^2 w_z}{\partial s_i^2}   \nonumber \\
&& {} + \sum_{i,j\in I_z,i<j} \rho_{ij}\sigma_i\sigma_js_is_j \frac{\partial^2 w_z}{\partial s_i\partial s_j} + \sum_{i\in I_z}\rho_i^T\sigma\pi\sigma_i xs_i \frac{\partial^2 w_z}{\partial x\partial s_i} \nonumber \\
&& {} + \sum_{i\in I_z} h^i_z(s) \left(w_{z^i}\left(t,x\left(1-\sum_{j=1}^N L_{ji}\pi^j\right),s^i\right)-  w_z\right) - \Gamma w_z. \label{new L}
\end{eqnarray}

We will show that $\widetilde{W}_z\leq \widetilde{V}_z$ for $\forall z\in I$ on $[0,T]\times [0,\infty)^{N_z+1}$ in the next few steps, thus we conclude $W_z\leq V_z$. We further define $\widetilde{F}$ function by
$$
\widetilde{F}_z\left(t,x,s,w,\nabla_{(t,x,s)}w_z,\nabla^2_{(x,s)}w_z\right) = - \sup_{\pi\in A} \widetilde{\mathcal{L}}^\pi w_z(t,x,s).
$$

\textbf{Step 2.} Define $\widetilde{V}^n_z:=\widetilde{V}_z+\frac{1}{n}\phi_z(t,x,s)$, where
$ \phi_z(t,x,s)=e^{-\lambda t}\left(1+x^{2\gamma}+\sum_{i\in I_z}s_i^{2\gamma}\right)$.
We claim that $\widetilde{V}^n$ is a viscosity supersolution to (\ref{new L}). Note that
\[ \mathcal{P}^{1,2,\ldots,2,-}\widetilde{V}^n_z(t,x,s) = \mathcal{P}^{1,2,\ldots,2,-}\widetilde{V}_z(t,x,s) + \frac{1}{n}\left(R^{\prime}, q^{\prime}, Q^{\prime}\right),  \]
where $R^{\prime} = -\lambda \phi_z$, $q^{\prime} = 2\gamma e^{-\lambda t}\big(x^{2\gamma-1},s_1^{2\gamma-1},\ldots,s_i^{2\gamma-1},\ldots,s_N^{2\gamma-1}\big)$ for $i\in I_z$ and 
\[
Q^{\prime} = 2\gamma(2\gamma-1)e^{-\lambda t}
 \begin{pmatrix}
  x^{2\gamma-2} & 0 & \ldots & 0  \\
  0 & s_1^{2\gamma-2} & \ldots & 0 \\
  \vdots & \vdots & \vdots & \vdots \\
  0 & 0 & \ldots & s_N^{2\gamma-2} 
 \end{pmatrix}.
\]

We have that for all $(R,q,Q)\in\mathcal{P}^{1,2,\ldots,2,-}\widetilde{V}^n_z(t,x,s)$, 
\[ \left(R-\frac{R^{\prime}}{n},q-\frac{q^{\prime}}{n},Q-\frac{Q^{\prime}}{n}\right)\in \mathcal{P}^{1,2,\ldots,2,-}\widetilde{V}_z(t,x,s). \]
Since $\widetilde{V}$ is a viscosity supersolution to (\ref{new L}), we have 
\[ \widetilde{F}_z\left( t,x,s,\widetilde{V},\left(R-\frac{R^{\prime}}{n},q-\frac{q^{\prime}}{n}\right),Q-\frac{Q^{\prime}}{n} \right)\geq 0 \] for $\forall z\in I$
by the equivalent definition of viscosity supersolution.
Using the inequality $\sup\{A-B\}\geq \sup\{A\}-\sup\{B\}$ and the boundedness of controls and coefficients, we have
\[
\widetilde{F}_z\left(t,x,s,\widetilde{V}^n,(R,q),Q\right) \geq \widetilde{F}_z\left( t,x,s,\widetilde{V},\left(R-\frac{R^{\prime}}{n},q-\frac{q^{\prime}}{n}\right),Q-\frac{Q^{\prime}}{n} \right) + \frac{1}{n}\left(\lambda + \Gamma + \sum_{i\in I_z}h_z^i(s) - K\right)\phi_z
\]
for a constant $K>0$.
Therefore, $\widetilde{F}_z\left(t,x,s,\widetilde{V}^n,(R,q),Q\right)\geq 0$ for a large enough $\lambda$, which implies that $\widetilde{V}^n$ is a viscosity supersolution of (\ref{new L}).

\textbf{Step 3.} We show that for all $n\geq 1$, it is $\widetilde{W}_z\leq \widetilde{V}^n_z$ for $z\in I$ on $[0,T)\times (0,\infty)^{N_z+1}$, and thus conclude that $\widetilde{W}\leq \widetilde{V}$. Fix $n\geq 1$ and define
\[ M_z:=\sup_{X\in [0,T)\times (0,\infty)^{N_z+1}}[\widetilde{W}_z(X)-\widetilde{V}^n_z(X)], \]
and
\[ M:=\max_{z\in I} M_z = M_{\bar z}, \]
where $X:=(t,x,s)$. We next show that $M\leq 0$. 
Suppose on the contrary that $M>0$, by the growth condition on $\widetilde{W}_{\bar z}$ and $\widetilde{V}_{\bar z}$ we have
\begin{equation*} \label{eq:x,s goes to infinity}
\lim_{x,s\rightarrow \infty}(\widetilde{W}_{\bar z}-\widetilde{V}^n_{\bar z})(t,x,s) = -\infty
\end{equation*}
for any $t\in[0,T)$.  By the terminal  and boundary conditions, we also have
\begin{equation*} \label{eq:t goes to T}
(\widetilde{W}_{\bar z}-\widetilde{V}^n_{\bar z})(T,x,s)\leq 0,\;
(\widetilde{W}_{\bar z}-\widetilde{V}^n_{\bar z})(t,0,s)\leq 0,\;
(\widetilde{W}_{\bar z}-\widetilde{V}^n_{\bar z})(t,x,0)\leq 0.
\end{equation*}
Note that here $s=0$ denotes $s_i=0$ for any $i\in I_z$.

Since $\widetilde{W}_{\bar z}-\widetilde{V}^n_{\bar z}$ is upper-semicontinuous and $M>0$, there exists some open bounded set $O\in [0,T)\times (0,\infty)^{N_{\bar z}+1}$ such that
\[ M=\max_{X\in O}[\widetilde{W}_{\bar z}(X)-\widetilde{V}^n_{\bar z}(X)]> 0. \]

We now use the doubling variable technique. For any fixed $\epsilon >0$, define  
\[ \Phi(X,Y) := \Phi_{\epsilon}(X,Y)=\widetilde{W}_{\bar z}(X)-\widetilde{V}^n_{\bar z}(Y)-\phi_1(X,Y),\]
where $\phi_1(X,Y):=\frac{1}{\epsilon}\lVert X-Y\rVert^2$. 
Note that $\Phi$ is upper-semicontinuous and hence achieves its maximum $\widetilde{M}=\widetilde{M}_{\epsilon}$ on the compact set $\bar{O}^2$ at $(\widetilde{X},\widetilde{Y})=(\widetilde{X}_{\epsilon},\widetilde{Y}_{\epsilon})$. We may write that, for all $\epsilon>0$,
\[
\begin{split}
M \leq \widetilde{M} & = \widetilde{W}_{\bar z}(\widetilde{X})-\widetilde{V}^n_{\bar z}(\widetilde{Y})-\phi_1(\widetilde{X},\widetilde{Y}) \leq \widetilde{W}_{\bar z}(\widetilde{X})-\widetilde{V}^n_{\bar z}(\widetilde{Y}).
\end{split}
\]
The sequence $(\widetilde{X},\widetilde{Y})$ converges, up to a subsequence, to some $(\hat{X},\hat{Y})\in \bar{O}^2$. Moreover, since $W_{\bar z}(\widetilde{X})-V^n_{\bar z}(\widetilde{Y})$ is upper bounded due to the upper-semicontinuity of $\widetilde{W}_{\bar z}$ and $-\widetilde{V}^n_{\bar z}$, we know $\phi_1(\widetilde{X},\widetilde{Y})$ is bounded, which implies $\hat{X}=\hat{Y}$.
Let  $\epsilon$ tend to 0 and take the $\limsup$, we get $M\leq \widetilde{W}_{\bar z}(\hat{X})-\widetilde{V}^n_{\bar z}(\hat{Y})\leq M$. Therefore, 
$\hat{X}=\hat{Y}\in O$ and $\phi_1(\widetilde{X},\widetilde{Y})\rightarrow 0$. 

\textbf{Step 4.} Since $(\widetilde{X},\widetilde{Y})$ converges to $(\hat{X},\hat{X})$ with $\hat{X}:=(\hat{t},\hat{x},\hat{s})\in O$, we may assume that for $\epsilon$ small enough, $(\widetilde{X},\widetilde{Y})$ lies in $O$.  We may write 
$\widetilde{X}:=(t_1,x_1,s_1)$ and $\widetilde{Y}:=(t_2,x_2,s_2)$. Then we have
\[
\nabla_{\widetilde{X}} \phi_1
=-\nabla_{\widetilde{Y}} \phi_1
={2\over \epsilon} (\widetilde{X}-\widetilde{Y}).
\]
Applying  Crandall-Ishii's lemma (see Crandall et al. (1992)), we have that  there exist $Q$ and $Q^{\prime}$ in $\mathcal{S}^{N_{\bar z}+1}$ such that
\begin{equation*}
\left(\nabla_{\widetilde{X}} \phi_1,Q\right)\in \bar{\mathcal{P}}^{1,2,\ldots,2,+}\widetilde{W}_{\bar z}(\widetilde{X}), \ \ \ \left(-\nabla_{\widetilde{Y}} \phi_1,Q^{\prime}\right)\in\bar{\mathcal{P}}^{1,2,\ldots,2,-}\widetilde{V}^n_{\bar z}(\widetilde{Y})
\end{equation*}
and the following matrix inequality holds in the non-negative definite sense:
\begin{equation*}
 \begin{pmatrix}
  Q & 0 \\
  0 & -Q^{\prime}
 \end{pmatrix}
\leq \frac{3}{\epsilon}
 \begin{pmatrix}
  I_{N_{\bar z}+1} & -I_{N_{\bar z}+1} \\
  -I_{N_{\bar z}+1} & I_{N_{\bar z}+1}
 \end{pmatrix}.
\end{equation*} 

By the viscosity subsolution (resp. supersolution) property of $\widetilde{W}$ (resp. $\widetilde{V}^n$), we have
\begin{equation} \label{eq:subsolution inequality}
\widetilde{F}_{\bar z}\left(t_1,x_1,s_1,\widetilde{W},\nabla_{\widetilde{X}} \phi_1,Q\right) \leq 0
\end{equation}
and
\begin{equation} \label{eq:supersolution inequality}
\widetilde{F}_{\bar z}\left(t_2,x_2,s_2,\widetilde{V}^n,-\nabla_{\widetilde{Y}} \phi_1,Q^{\prime}\right) \geq 0.
\end{equation}

 Subtracting (\ref{eq:subsolution inequality}) from (\ref{eq:supersolution inequality}), using the fact that the difference of the supreme is less than the supreme of the difference, we obtain
\[
\Gamma\left(\widetilde{W}_{\bar z}(\widetilde{X})-\widetilde{V}^n_{\bar z}(\widetilde{Y})\right) + \sum_{i\in I_{\bar z}}\left(h_{\bar z}^i(s_1)\widetilde{W}_{\bar z}(\widetilde{X}) - h_{\bar z}^i(s_2)\widetilde{V}^n_{\bar z}(\widetilde{Y})\right) \leq \sup_{\pi\in A}\Big \{ J_1(\pi) + J_2(\pi) + J_3(\pi) \Big \},
\]
where
$$
J_1(\pi) = (r+\theta^T\pi)\frac{2(x_1-x_2)^2}{\epsilon} + \sum_{i\in I_{\bar z}}\mu_i\frac{2(s_{1i}-s_{2i})^2}{\epsilon},
$$
$$
J_2(\pi) = \sum_{i\in I_{\bar z}}\left(h_{\bar z}^i(s_1)\widetilde{W}_{\bar{z}^i}\left( t_1,x_1\left(1-\sum_{j=1}^N L_{ji}\pi^j\right),s_1^i \right) - h_{\bar z}^i(s_2)\widetilde{V}^n_{\bar{z}^i}\left( t_2,x_2\left(1-\sum_{j=1}^N L_{ji}\pi^j\right),s_2^i \right)\right),
$$
and
\begin{eqnarray}
J_3(\pi) &= & \frac{1}{2}\pi^T\Sigma\pi \left(x_1^2 Q_{1,1} - x_2^2 Q_{1,1}^{\prime}\right) + \frac{1}{2}\sum_{i\in I_{\bar z}} \sigma_i^2\left(s_{1i}^2 Q_{k_i,k_i} - s_{2i}^2 Q_{k_i,k_i}^{\prime}\right) \nonumber \\
&& {} + \sum_{i,j\in I_{\bar z},i<j} \rho_{ij}\sigma_i\sigma_j\left(s_{1i}s_{1j}Q_{k_i,k_j} - s_{2i}s_{2j} Q_{k_i,k_j}^{\prime}\right)  + \sum_{i\in I_{\bar z}}\rho_i^T\sigma\pi\sigma_i \left(x_1s_{1i} Q_{1,k_i} - x_2s_{2i} Q_{1,k_i}^{\prime}\right). \nonumber 
\end{eqnarray}

Since $\phi_1(\widetilde{X},\widetilde{Y})\rightarrow 0$, we can derive $\limsup_{\epsilon\rightarrow 0} J_1(\pi) = 0$ for any $\pi$. By the definition of $M$, we have $\limsup_{\epsilon\rightarrow 0} J_2(\pi) \leq \sum_{i\in I_{\bar z}}h_{\bar z}^i(\hat{s})M$ for any $\pi$. By the structure condition and Crandall Ishii's inequality, we have 
\[
J_3(\pi) \leq \frac{K}{\epsilon}\left(|x_1-x_2|^2 + \sum_{i\in I_{\bar z}} |s_{1i}-s_{2i}|^2 \right).
\]
Thus we can derive that $\limsup_{\epsilon\rightarrow 0} J_3(\pi) \leq 0$ for any $\pi$. Therefore 
\[
\limsup_{\epsilon \rightarrow 0} \left(\Gamma\left(\widetilde{W}_{\bar z}(\widetilde{X})-\widetilde{V}^n_{\bar z}(\widetilde{Y})\right) + \sum_{i\in I_{\bar z}}\left(h_{\bar z}^i(s_1)\widetilde{W}_{\bar z}(\widetilde{X}) - h_{\bar z}^i(s_2)\widetilde{V}^n_{\bar z}(\widetilde{Y})\right) \right) = \Gamma M + \sum_{i\in I_{\bar z}}h_{\bar z}^i(\hat{s})M \leq \sum_{i\in I_{\bar z}}h_{\bar z}^i(\hat{s})M.
\]
Since $\Gamma>0$, we have 
$M \leq 0$, 
which is a contradiction to the assumption that $M>0$.
We conclude that $M\leq 0$, which implies $W_z\leq V_z$ for $\forall z\in I$ on $[0,T)\times (0,\infty)^{N_z+1}$.
\end{proof}

\section{Conclusions}
In this paper we  consider a utility maximization problem with looping  contagion risk.
 We  assume that the default intensity of one company depends on the stock prices of other companies and the default of one company induces immediate drops in the stock prices of the other surviving companies.  In addition to the verification theorem, we   prove 
 the  value function is the unique viscosity solution of the  HJB equation system.  We also compare and analyse the statistical distributions of terminal wealth of log utility based on two optimal strategies,  one using the full information of intensity process, the other a proxy constant intensity process.   Our numerical tests  show that, statistically, using trading strategies based on stock price dependant intensities would  achieve higher return on average, especially when the difference of the stock dependent intensity and the proxy constant intensity is big, but could also be more volatile in extreme scenarios. There remain many open questions in utility maximization with contagion risk, for example, the BSDE simulation method for power utility.   We leave these and other questions to future research.

\bigskip\noindent
{\bf Acknowledgements}. The authors are very grateful to  anonymous reviewers whose constructive comments and suggestions have helped to improve the paper of the previous version.


\end{document}